\documentclass{llncs}
\usepackage{amssymb}
\usepackage{graphicx}

\usepackage{ntheorem}
\theoremseparator{.}
\newtheorem{obs}[theorem]{Observation}
\newtheorem{lem}[theorem]{Lemma}
\newtheorem{cor}[theorem]{Corollary}
\newtheorem{prop}[theorem]{Proposition}
\newtheorem{reducerule}{Reduction Rule}
\newtheorem{branchrule}[reducerule]{Branching Rule}
\newtheorem{reducebranchrule}[reducerule]{Reduction$\setminus$Branching Rule}

\newcommand{\negA}{\vspace{-0.05in}}
\newcommand{\negB}{\vspace{-0.1in}}

\newcommand{\mysubsection}[1]{\negB\subsection{#1}\negA}

\newcommand{\myparagraph}[1]{\par\smallskip\par\noindent{\bf{}#1:~}}
\newcommand{\etal}{{\em et al.}}
\newcommand{\comment}[1]{}

\usepackage{algorithm}
\usepackage{algorithmic}
\newcommand{\alg}[1]{\mbox{\sf #1}}  

\usepackage{url}
\urldef{\mails}\path|{hadas, meizeh}@cs.technion.ac.il|     

\begin{document}

\mainmatter

\title{A Multivariate Framework for\\Weighted FPT Algorithms}

\author{Hadas Shachnai \and Meirav Zehavi}

\institute{Department of Computer Science, Technion, Haifa 32000, Israel\\
\mails}

\maketitle


\begin{abstract}
We introduce a novel multivariate approach for solving weighted parameterized problems. 
In our model, given an instance of size $n$ of a minimization (maximization) problem, and
a parameter $W \geq 1$, we seek a solution of weight at most (or at least) $W$. We use our general framework
to obtain efficient algorithms for such fundamental graph problems as {\sc Vertex Cover}, {\sc 3-Hitting Set}, 
{\sc Edge Dominating Set} and {\sc Max Internal Out-Branching}. 
The best known algorithms for 
these problems admit running times of the form $c^W n^{O(1)}$, for some constant $c>1$. We improve these running
 times to $c^s n^{O(1)}$, where $s\leq W$ is the minimum size of a solution of weight at most (at least) $W$. 
If no such solution exists, $s=\min\{W,m\}$, where $m$ is the maximum size of a solution. Clearly, $s$ can be 
substantially smaller than $W\!$. In particular, the running times of our algorithms are (almost) the same as the best 
known $O^*$ running times for the unweighted variants. Thus, we solve the weighted versions~of

\smallskip
\renewcommand{\labelitemi}{$\bullet$}
\renewcommand{\labelitemii}{$-$}
\begin{itemize}
\item {\sc Vertex Cover} in $1.381^s n^{O(1)}$ time and $n^{O(1)}$ space.
\item {\sc 3-Hitting Set} in $2.168^s n^{O(1)}$ time and $n^{O(1)}$ space.
\item {\sc Edge Dominating Set} in $2.315^s n^{O(1)}$ time and $n^{O(1)}$ space.
\item {\sc Max Internal Out-Branching} in $6.855^s n^{O(1)}$ time and space.
\end{itemize}
\smallskip

We further show that {\sc Weighted Vertex Cover} and {\sc Weighted Edge Dominating Set} admit fast algorithms whose 
running times are of the form $c^t n^{O(1)}$, where $t \leq s$ is the minimum size of a solution.
\end{abstract}

\section{Introduction}

Many fundamental problems in graph theory are NP-hard already on {\em unweighted} graphs. This wide class includes, among others, {\sc Vertex Cover}, {\sc 3-Hitting Set}, {\sc Edge Dominating Set} and {\sc Max Internal Out-Branching}~\cite{subgraphisoneg}. Fast existing parameterized algorithms for these problems, which often exploit the structural properties of the underlying graph, cannot be naturally extended to handle weighted instances. Thus, solving efficiently weighted graph problems has remained among the outstanding open questions in parameterized complexity, as excellently phrased by Hajiaghayi~\cite{Dagstuhl09511}:
\negA
\begin{quote}
``Most fixed-parameter algorithms for parameterized problems are inherently about {\em unweighted} graphs. Of course, we could add integer weights to the problem, but this can lead to a huge increase in the parameter. Can we devise fixed-parameter algorithms for weighted graphs that have less severe dependence on weights? Is there a nice framework for designing fixed-parameter algorithms on weighted graphs?''
\end{quote}
\negA

We answer these questions affirmatively, by developing a multivariate framework for solving weighted parameterized problems. We use this framework to obtain efficient algorithms for the following fundamental graph problems.

\myparagraph{Weighted Vertex Cover (WVC)} Given a graph $G=(V,E)$, a weight function $w: V\rightarrow \mathbb{R}^{+}$, and a parameter $W\in\mathbb{R}^{+}$, find a vertex cover $U\subseteq V$ (i.e., every edge in $E$ has an endpoint in $U$) of weight at most $W$ (if one exists).

\myparagraph{Weighted 3-Hitting Set (W3HS)} Given a 3-uniform hypergraph $G=(V,E)$, a weight function $w: V\rightarrow \mathbb{R}^{+}$, and a parameter $W\in\mathbb{R}^{+}$, find a hitting set $U\subseteq V$ (i.e., every hyperedge in $E$ has an endpoint in $U$) of weight at most $W$ (if one exists).

\myparagraph{Weighted Edge Dominating Set (WEDS)} Given a graph $G=(V,E)$, a weight function $w: E\rightarrow \mathbb{R}^{+}$, and a parameter $W\in\mathbb{R}^{+}$, find an edge dominating set $U\subseteq E$ (i.e., every edge in $E$ touches an endpoint of an edge in $U$) of weight at most $W$ (if one exists).

\myparagraph{Weighted Max Internal Out-Branching (WIOB)} Given a directed graph $G=(V,E)$, a weight function $w: V\rightarrow \mathbb{R}^{+}$, and a parameter $W\in\mathbb{R}^{+}$, find an out-branching of $G$ (i.e., a spanning tree having exactly one vertex of in-degree 0) having internal vertices of total weight at least $W$ (if one exists).

\smallskip
Parameterized algorithms solve NP-hard problems by confining the combinatorial explosion to a parameter $k$. 
More precisely, a problem is {\em fixed-parameter tractable (FPT)} with respect to a parameter $k$ if it can be solved in time $O^*\!(f(k))$ for some function $f$, where $O^*$ hides factors polynomial in the input size $n$.

Existing FPT algorithms for the above problems have running~times~of~the form $O^*\!(c^W\!)$. Using our framework,
 we obtain {\em faster} algorithms, whose running times are 
of the form $O^*\!(c^s)$, where $s\!\leq\! W$ is the minimum size of a solution of weight at most (at least) $W$. If no such solution exists, $s\!=\!\min\{W\!,m\}$, where $m$ is the maximum size of a solution (for the unweighted 
version).\footnote{We note that obtaining {\em slow} running times of this form is simple; the challenge lies in having the bases
the same as in the previous best known $O^*\!(c^W\!)$ running times.} Clearly, $s$ can be significantly 
smaller than $W$. Moreover, almost all of the bases in our $O^*\!(c^s)$ running times {\em improve upon} those in the previous
 best known $O^*\!(c^W\!)$ running times for our problems. We complement these results by developing algorithms for {\sc Weighted Vertex Cover} and {\sc Weighted Edge Dominating Set} parameterized by $t\!\leq\! s$, the minimum size of a solution (for the~unweighted~version).

\subsection{Previous Work}\label{section:priorwork}

Our problems are well known in graph theory and combinatorial optimization. They were also extensively studied in the area of parameterized complexity. We mention below known FPT results for their unweighted and weighted variants, parameterized by $t$ and $W$, respectively.

\myparagraph{Vertex Cover} {\sc VC} is one of the first problems shown to be FPT. In the past two decades, it enjoyed a race towards obtaining the fastest FPT algorithm (see Appendix \ref{app:priorwork}). The best FPT algorithm,
 due to Chen~\etal~\cite{vc2010}, has running time $O^*(1.274^t)$. In a similar race, focusing on graphs of bounded degree 3 (see Appendix \ref{app:priorwork}), the current winner is an algorithm of
 Issac~\etal~\cite{3vc2013}, whose running time is $O^*(1.153^t)$. For {\sc WVC}, Niedermeier~\etal~\cite{wvc2003} proposed an algorithm of $O^*(1.396^W)$ time and polynomial space, and an algorithm of $O^*(1.379^W)$ time and $O^*(1.363^W)$ space. Subsequently,
 Fomin~\etal~\cite{wvc2006} presented an algorithm of $O^*(1.357^W)$ time and space. An alternative algorithm, using $O^*(1.381^W)$ time and $O^*(1.26^W)$ space, is given in \cite{wvc2009}.

\myparagraph{3-Hitting Set} Several papers study FPT algorithms for {\sc 3HS} (see \cite{hs2003,hs2010a,hs2010b,hs2007}). The best such algorithm, by Wahlstr$\ddot{\mathrm{o}}$m \cite{hs2007}, has running time $O^*(2.076^t)$. For {\sc W3HS}, Fernau \cite{whs2010} gave an algorithm which runs in time $O^*(2.247^W)$ and uses polynomial space.

\myparagraph{Edge Dominating Set} FPT algorithms for {\sc EDS} are given in \cite{weds2006,weds2010,eds2013}, and the papers \cite{3eds2010,3eds2013} present such algorithms for graphs of bounded degree 3. The best known algorithm for general graphs, 
due to Xiao~\etal~\cite{eds2013}, has running time $O^*(2.315^t)$, and for graphs of bounded degree 3, the current best algorithm, 
due to Xiao~\etal~\cite{3eds2013}, has running time $O^*(2.148^t)$. FPT algorithms for {\sc WEDS} are given in \cite{weds2006,wvc2009,weds2010}. The best algorithm, due to Binkele-Raible~\etal~\cite{weds2010}, runs in time $O^*(2.382^W)$ and uses polynomial space.

\myparagraph{Max Internal Out-Branching} Although FPT algorithms for minimization problems are more common than those for maximization problems (see \cite{newdfsbook}), {\sc IOB} was extensively studied in this area (see Appendix \ref{app:priorwork}). The previous best algorithms run in time $O^*(6.855^t)$ \cite{corrrepresentative}, and
 in randomized time $O^*(4^t)$ \cite{thesis11,ipec13}. The weighted version, {\sc WIOB}, was studied in the
area of approximation algorithms  (see \cite{IOBapprox1,IOBapprox2}); 
however, to the best of our knowledge, its parameterized complexity is studied here for the first time.
\smallskip

We note that well-known tools, such as the color coding technique \cite{colorcoding}, can be used to obtain elegant FPT algorithms for some classic weighted graph problems (see, e.g., \cite{representative,colorcodingeng,qnet}). 
Recently, Cygan~\etal~\cite{bisection} introduced a novel form of tree-decomposition to develop an FPT algorithm for minimum weighted graph bisection. Yet, for many other problems, including those studied in this paper, these tools are not known to be useful. We further elaborate in Section \ref{section:technique} on the limitations of known techniques in solving weighted graph problems.

\subsection{Our Results}

We introduce a novel multivariate approach for solving weighted parameterized problems. Our framework 
yields fast algorithms whose running times are of the form $O^*(c^s)$. We demonstrate its usefulness for the following 
problems.
 
\renewcommand{\labelitemi}{$\bullet$}
\renewcommand{\labelitemii}{$-$}
\begin{itemize}
\item {\sc WVC}: We give an algorithm that uses $O^*(1.381^s)$ time and polynomial space, or $O^*(1.363^s)$ time and space, 
complemented by an algorithm that uses $O^*(1.443^t)$ time and polynomial space. For graphs of bounded degree 3, this algorithm runs in time $O^*(1.415^t)$.\footnote{We also give (in Appendix \ref{section:wvc2}) an $O^*(1.347^W)$ time algorithm for {\sc WVC}.}

\item {\sc W3HS}: We develop an algorithm which uses $O^*(2.168^s)$ time and polynomial space, complemented by an algorithm which uses $O^*(1.381^{s-t}2.381^t)$ time and polynomial space, or $O^*(1.363^{s-t}2.363^t)$ time and $O^*(1.363^s)$ space.

\item {\sc WEDS}: We give an algorithm which uses $O^*(2.315^s)$ time and polynomial space, complemented by an $O^*(3^t)$ time and polynomial space algorithm.

\item {\sc WIOB}: We present an algorithm that has time and space complexities $O^*(6.855^s)$, 
or randomized time and space $O^*(4^sW)$.
\end{itemize}

\begin{table}[center]
\centering
\begin{tabular}{|l|c|c|c|c|l|}
\hline
Problem   & Unweighted                      & Parameter $W$ & Parameter $s$ & Parameter $(t\!+\!s)$ & Comments \\\hline\hline
{\sc WVC}  & $O^*\!(1.274^t)$ \cite{vc2010} & $O^*\!(1.396^W\!)$ \cite{wvc2003} & $\bf O^*\!(1.381^s)$ & $\bf O^*\!(1.443^t)$ & $O^*\!(1)\!$ space \\
					& $\cdot$  												& $O^*\!(1.357^W\!)$  \cite{wvc2006} & $\bf O^*\!(1.363^s)$ & $\bf \cdot$ & \\        
          & $O^*\!(1.153^t)$ \cite{3vc2013} & $\cdot$  & $\bf \cdot$  & $\bf O^*\!(1.415^t)$ & $\Delta=3$ \\\hline                    
{\sc W3HS}& $O^*\!(2.076^t)$ \cite{hs2007}  & $O^*\!(2.247^W\!)$ \cite{whs2010} & $\bf O^*\!(2.168^s)$ & $\bf O^*\!(1.363^{s\!-\!t}2.363^t)$  &  \\\hline
{\sc WEDS}& $O^*\!(2.315^t)$ \cite{eds2013} & $O^*\!(2.382^W\!)$ \cite{weds2010} & $\bf O^*\!(2.315^s)$ & $\bf O^*\!(3^t)$ & \\
          & $O^*\!(2.148^t)$ \cite{3eds2013}& $\cdot$ & $\bf \cdot$ & $\bf \cdot$ & $\Delta=3$ \\\hline                  
{\sc WIOB}& $O^*\!(6.855^t)$ \cite{corrrepresentative}& --- & $\bf O^*\!(6.855^s\!)$ & --- & \\\hline	
\end{tabular}\medskip
\caption{Known results for {\sc WVC}, {\sc W3HS}, {\sc WEDS} and {\sc WIOB}, parameterized by $t$, $W$ and $s$.}
\label{tab:knownresults}
\end{table}

Table \ref{tab:knownresults} summarizes the known results for our problems. Results given in this paper are shown in boldface. Entries marked with {{$\cdot$} follow by inference from the first entry in the same cell. As shown in Table \ref{tab:knownresults}, our results imply that even if $W$ is large, our problems can be solved efficiently, i.e., in times that are comparable to those required for solving their unweighted counterparts. Furthermore, most of the bases in our $O^*(c^s)$ running times are smaller than the bases in the corresponding known $O^*(c^W)$ running times. One may view such fast running times as somewhat surprising, since {\sc WVC}, a key player in deriving our results, seems inherently more difficult than {\sc VC}. Indeed, while {\sc VC} admits a kernel of size $2t$, the smallest known kernel for {\sc WVC} is of size $2W$ \cite{vc2001,wvcker2008}. In fact, as shown in \cite{wvcker2013}, {\sc WVC} 
does not admit a polynomial kernel when parameterized~by~$t$.

\myparagraph{Technical Contribution}
A critical feature of our framework is that 
it allows an algorithm to 
``fail" in certain executions, e.g.,
to return NIL even if there exists a solution of weight at most (at least) $W$ for
the given input (see Section \ref{section:technique}). We obtain improved running times for our algorithms by 
exploiting this feature, along with
an array of sophisticated tools 
 for tackling our problems. Specifically, in solving {\em minimization} problems, we show how the framework
 can be used to eliminate branching steps along the construction of bounded search trees, thus decreasing
 the overall running time. In solving {\sc WIOB}, we reduce a given problem instance to an instance of an auxiliary
 problem, called {\sc Weighted $k$-ITree}, for which we obtain an initial solution (see Appendix \ref{section:wiob}). 
This solution is then transformed into a solution for the original instance. Allowing ``failures'' for the algorithms simplifies 
the subroutine which solves {\sc Weighted $k$-ITree}, since we do not need to ensure that the initial solution is not ``too big''. 
Again, this results in improved~running~times.

Furthermore, our framework makes non-standard use of
the classic {\em bounded search trees} technique. Indeed, the analysis
of an algorithm based on the technique relies on bounds attained by tracking the underlying 
input parameter, and the corresponding {\em branching vectors} of the algorithms (see Section \ref{section:preliminaries}).
In deriving our results, we track the value of the 
{\em weight parameter} $W$, but analyze the branching vectors with respect to a {\em special size parameter} $k$.
Our algorithms may base their output on the value of $W$ only (i.e., ignore $k$),\footnote{See, e.g., Rule \ref{rule2:hsdeg2} 
in the algorithm for W3HS in Appendix \ref{section:whs1}.} or may 
decrease $k$ by {\em less} than its actual decrease in the instance.\footnote{See, e.g., Rules \ref{rule:deg1} 
and \ref{rule:vctriangle} in Section \ref{section:wvc1}.} 

\myparagraph{Organization}
In Section \ref{section:preliminaries}, we give some definitions and notation, including an overview of the bounded search 
trees technique. Section \ref{section:technique} presents our general multivariate framework. In Section \ref{section:wvc1},
 we demonstrate the usefulness of our framework by developing an $O^*(1.381^s)$ time and polynomial space algorithm for
{\sc WVC}. Due to lack of space, the rest of the applications are relegated to the Appendix (also given in \cite{corrweighted}). 
In particular, Appendix \ref{app:wvc} contains additional algorithms and a hardness result related to {\sc WVC}, 
and Appendices \ref{section:w3hs}, \ref{sec:WEDS} and \ref{section:wiob} contain 
algorithms for {\sc W3HS}, {\sc WEDS} and {\sc WIOB}, respectively.

\section{Preliminaries}\label{section:preliminaries}

\myparagraph{Definitions and Notation} Given a (hyper)graph $G=(V,E)$ and a vertex $v\in V$, let $N(v)$ denote the set
of neighbors of $v$; $E(v)$ denotes the set of edges adjacent to $v$. The {\em degree} of $v$ is $|E(v)|$ (which, for hypergraphs, 
may not be equal to $|N(v)|$). Recall that a {\em leaf} is a degree-1 vertex. Given a subgraph $H$ of $G$, let $V(H)$ and $E(H)$ denote its vertex set and edge set, respectively. For a subset $U\subseteq V$, let $N(U)=\bigcup_{v\in U}N(v)$, and $E(U)=\bigcup_{v\in U}E(v)$. 
Also, we denote by $G[U]$ the subgraph of $G$ induced by $U$ (if $G$ is a hypergraph, $v,u\in U$ and $r\in V\setminus U$ such that $\{v,u,r\}\in E$, then $\{v,u\}\in E(G[U])$). Given a set $S$ and a weight function $w: S\rightarrow \mathbb{R}$, the total weight of $S$ is given by $w(S) = \sum_{s\in S} w(s)$. Finally, we say that a (hyper)edge $e\in E$ containing exactly $d$ vertices is a {\em $d$-edge}.

In deriving our results, we assume that $W$ and element weights are at least $1$ (indeed, this ensures fixed-parameter 
tractability with respect to $W$ (see, e.g.,~\cite{wvc2003})). 

\myparagraph{Bounded Search Trees} The {\em bounded search trees} technique is fundamental in the design of recursive FPT algorithms (see, e.g., \cite{newdfsbook}).
 Informally, in applying this technique, one defines a list of rules. Each rule is of the form \alg{Rule X. [condition] action}, 
where \alg{X} is the number of the rule in the list. At each recursive call (i.e., a node in the search tree), the algorithm 
performs the action of the first rule whose condition is satisfied. If, by performing an action, the algorithm recursively
calls itself at least twice, the rule is a {\em branching rule}; otherwise, it is a {\em reduction rule}. 
We only consider polynomial time actions that increase neither the parameter nor the size of the instance, 
and decrease at least one of them. We give in Appendix \ref{app:boundedsearch} detailed examples, showing how to
solve {\sc VC} and {\sc WVC} using the technique. 

The running time of the algorithm which uses bounded search trees can be analyzed as follows. 
Suppose that the algorithm executes a branching rule
which has $\ell$ branching options (each leading to a recursive call with the corresponding parameter value),
such that in the $i^\mathrm{th}$ branch option, the current value of the parameter decreases by $b_i$. Then, $(b_1,b_2,\ldots,b_{\ell})$ is called the {\em branching vector} of this rule. We say that $\alpha$ is the {\em root} of $(b_1,b_2,\ldots,b_{\ell})$ if it is the (unique) positive real root of $x^{b^*} = x^{b^*-b_1} + x^{b^*-b_2} + \ldots + x^{b^*-b_{\ell}}$, where $b^* = \max\{b_1,b_2,\ldots,b_{\ell}\}$. 
If $r\!>\!0$ is the initial value of the parameter, and the algorithm (a) returns a result when (or before) the parameter is negative, and (b) only executes branching rules whose roots are bounded by 
a constant $c >0$, then its running time is bounded by $O^*(c^r)$.
\comment{
In this paper, we combine our multivariate framework (see Section \ref{section:technique}) with the bounded search
 tree technique, thus introducing a novel way of applying this well-known technique. We track the value of the 
{\em weight parameter} $W$, but analyze the branching vectors with respect to a {\em special size parameter} $k$:
 an algorithm can ignore $k$ in some cases where it returns an answer (i.e., the answer may depend only on $W$; see, e.g.,
 Appendix \ref{section:wvcbip}, Rule \ref{rule2:hsdeg2} in Appendix \ref{section:whs1}, or Appendix \ref{section:edsvc}), 
and it can decrease $k$ by {\em less} than its actual decrease in the instance (see, e.g., Rules \ref{rule:deg1} 
and \ref{rule:vctriangle} in Section \ref{section:wvc1}, Rule \ref{rule:notcontrule} in the proof of 
Theorem \ref{theorem:wvc2}, or Appendix \ref{app:boundedsearch}).
}
\section{A General Multivariate Framework}\label{section:technique}

In our framework, a problem parameterized by the solution weight is solved by adding a special size parameter. Formally,
given a problem instance, and a weight parameter $W >1$, we add an integer parameter $0 < k\leq W$. We then seek a solution of weight at most (at least) $W$.
 The crux of the framework is
 in allowing our algorithms to ``fail'' in certain cases. This enables to substantially improve running times, while maintaining the correctness of the returned solutions. Specifically, our algorithms satisfy the following properties. Given $W$~and~$k$,

\renewcommand{\labelitemi}{$\bullet$}
\renewcommand{\labelitemii}{$-$}
\begin{enumerate}
\item [$(i)$] If there exists a solution of weight at most (at least) $W$, and size at most $k$, return~a~solution of weight at most (at least) $W$. The size of the returned solution may be larger than $k$.

\item[$(ii)$] Otherwise, return NIL, or a solution of weight at most (at least) $W$.
\end{enumerate}
Clearly, the correctness of the solution can be maintained by iterating the above step, until we reach a value of $k$ for which $(i)$ is satisfied and the algorithm terminates with ``success''. 
\comment{
The decrease in running time is achieved by combining our framework with an array of sophisticated tools used
 for handling our problems. Specifically, in solving {\em minimization} problems, we show how the framework
 can be used to eliminate branching steps along the construction of bounded search trees, thus decreasing
 the overall running time. In solving {\sc WIOB}, we reduce a given problem instance to an instance of an auxiliary
 problem, called {\sc Weighted $k$-ITree}, for which we obtain an initial solution (see Appendix \ref{section:wiob}). 
This solution is then transformed into a solution for the original instance. Allowing ``failures'' for the algorithms simplifies 
the subroutine which solves {\sc Weighted $k$-ITree}, since we do not need to ensure that the initial solution is not ``too big''. 
Again, this results in improved~running~times.
}
Using our framework, we  solve the following problems.

\smallskip

\myparagraph{$k$-WVC}Given an instance of {\sc WVC}, along with a parameter $k\!\in\!\mathbb{N}$, satisfy the following.~If there is a vertex cover of weight at most $W$ and size at most $k$, return a vertex cover of weight at most $W$; otherwise, return NIL, or a vertex cover of weight at most $W$.

\myparagraph{$k$-W3HS}Given an instance of {\sc W3HS}, along with a parameter $k\!\in\!\mathbb{N}$, satisfy the following.~If there is a hitting set of weight at most $W$ and size at most $k$, return a hitting set of weight at most $W$; otherwise, return NIL or a hitting set of weight at most $W$.

\myparagraph{$k$-WEDS}Given an instance of {\sc WEDS}, along with a parameter $k\!\in\!\mathbb{N}$, satisfy the following.~If there is an edge dominating set of weight at most $W$ and size at most $k$, return an edge dominating set of weight at most $W$; otherwise, return NIL or an edge dominating set of weight at most $W$.

\myparagraph{$k$-WIOB}Given an instance of {\sc WIOB}, along with a parameter $k\!<\!W$, satisfy the following.~If there is an out-branching having a set of internal vertices of total weight at least $W$ and cardinality at most $k$, return an out-branching with internal vertices of total weight at least $W\!;$ otherwise, return NIL or an out-branching with internal vertices of total weight~at~least~$W\!.$\footnote{If $k\geq W$, assume that {\sc $k$-WIOB} is simply {\sc WIOB}.}

\medskip

We develop FPT algorithms for the above variants, which are then used to solve the original problems. Initially, $k=1$. We increase this value iteratively, until either $k=\min\{W,m\}$, or a solution of weight at most (at least) $W$ is found, where $m$ is the maximum size of a solution. Clearly, for {\sc WVC} and {\sc W3HS}, $m=|V|$; for {\sc WEDS}, $m=|E|$; and for {\sc WIOB}, $m$ is the maximum number of internal vertices of a spanning tree of $G$. For {\sc WIOB}, to ensure that $s\! \leq\!\min\{W\!,m\}$, we proceed as follows. Initially, we solve {\sc $1$-WIOB}. While the algorithm returns NIL, before incrementing the value of $k$, we solve {\sc IOB}, in which we seek an out-branching having at least $(k+1)$ internal vertices (using \cite{corrrepresentative,thesis11}). Our algorithm solves {\sc $(k+1)$-WIOB} only if the returned answer $\neq$ NIL and $k+1\!\leq\! W$.

We note that some weighted variants of parameterized problems were studied in the following {\em restricted} form. Given a problem instance, along with the parameters $W \geq 1$ and $k\in\mathbb{N}$, find a solution of weight at most (at least) $W$ and size at most $k$; if such a solution does not exist, return NIL (see, e.g., \cite{wfvs2008,wcvc2012}). Clearly, an algorithm for this variant can be used to obtain running time of the form $O^*(c^s)$ for the original weighted instance. However, the efficiency of our algorithms crucially relies on the flexible use of the parameter $k$. In particular (as shown in Appendix \ref{section:wvcbip}), for some of the problems, the restricted form becomes NP-hard already on easy classes of graphs, as opposed to the above problems, which remain polynomial time solvable on such graphs.
\section{An $O^*(1.381^s)$ Time Algorithm for {\sc WVC}}\label{section:wvc1}

In this section, we present our first algorithm, \alg{WVC-Alg}. This algorithm employs the bounded search tree technique, described in Section \ref{section:preliminaries}. It builds upon rules used by the $O^*(1.396^W)$ time and polynomial space algorithm for {\sc WVC} proposed in \cite{wvc2003}. However, we also introduce new rules, including, among others, 
reduction rules that manipulate the weights of the vertices in the input graph. This allows us to easily and efficiently eliminate leaves and certain triangles (see Rules \ref{rule:deg1} and \ref{rule:vctriangle}). Thus, we obtain the following.

\begin{theorem}
\alg{WVC-Alg} solves {\sc $k$-WVC} in $O^*(1.381^k)$ time and polynomial space.
\end{theorem}

By the discussion in Section \ref{section:technique}, this implies the desired result:

\begin{cor}
{\sc WVC} can be solved in $O^*(1.381^s)$ time and polynomial space.
\end{cor}

Next, we present each rule within a call \alg{WVC-Alg}($G=(V,E),w: V\!\rightarrow\! \mathbb{R}^{\geq 0},W,k$).
Initially, \alg{WVC-Alg} is called with a weight function $w$,
 whose image lies in $\mathbb{R}^{\geq 1}$. 
After presenting a rule, we argue its correctness. For each branching rule, we also give the root of the corresponding
branching vector (with respect to $k$). Since the largest root we shall get is bounded by 1.381, and the algorithm stops if $k<0$, we have the desired running time. 

\begin{reducerule}\label{rule:vckneg}
{\normalfont [$\min\{W,k\}<0$]
Return NIL.}
\end{reducerule}

{\noindent If $\min\{W,k\}<0$, there is no vertex cover of weight at most $W$ and size at~most~$k$.}

\begin{reducerule}
{\normalfont [$E=\emptyset$]
Return $\emptyset$.}
\end{reducerule}

{\noindent Since $E=\emptyset$, an empty set is a vertex cover.}

\begin{reducerule}\label{rule:deg2component}
{\normalfont [There is a connected component $H$ with at most one vertex of degree at least 3, where $|E(H)|\geq 1$]
Use dynamic programming to compute a minimum-weight vertex cover $U$ of $H$ (see \cite{wvc2003}). Return \alg{WVC-Alg}$(G[V\setminus V(H)], w, W\!-w(U), k-1)\ \cup\ U$.\footnote{We assume that adding elements to NIL results in NIL.}}
\end{reducerule}

{\noindent Since $H$ is a connected component, any minimum-weight vertex cover of $G$ consists of a vertex cover of $H$ of weight $w(U)$, and a minimum-weight vertex cover of $G[V\setminus V(H)]$. Furthermore, any vertex cover of $G$ contains a vertex cover of $H$ of size at least 1. Therefore, we return a solution as required: if there is a solution of size at most $k$ and weight at most $W$, we return a solution of weight at most $W$, and if there is no solution of weight at most $W$, we~return~NIL.}

\begin{reducerule}\label{rule:concomponent}
{\normalfont [There is a connected component $H$ such that $|V(H)|\leq 100$ and $|E(H)|\geq 1$]
Use brute-force to compute a minimum-weight vertex cover $U$ of $H$. Return \alg{WVC-Alg}$(G[V\setminus V(H)], w, W\!-w(U), k-1)\ \cup\ U$.}
\end{reducerule}

{\noindent 
The correctness of the rule follows from the same arguments as given for Rule \ref{rule:deg2component}. 
The next rule, among other rules, clarifies the necessity of Rule \ref{rule:concomponent}, and, in particular, the choice of the value 100.\footnote{Choosing a smaller value is possible, but it
is unnecessary and complicates~the~proof.}}

\begin{branchrule}\label{rule:deg4}
{\normalfont [There is a vertex of degree at least 4, or all vertices have degree 3 or 0]
Let $v$ be a vertex of maximum degree.
\begin{enumerate}
\item If the result of \alg{WVC-Alg}($G[V\setminus \{v\}],w,W\!-w(v),k-1$) is not NIL: Return it along with~$v$.
\item Else: Return \alg{WVC-Alg}($G[V\setminus N(v)],w,W\!-w(N(v)),k-\max\{|N(v)|,4\}$), along with~$N(v)$.
\end{enumerate}}
\end{branchrule}

{\noindent This branching is exhaustive. If the degree of $v$ is at least 4, the rule is clearly correct; else, the degree of any vertex in $G$ is 3 or 0. 
Then, we need to argue that decreasing $k$ by 4 in the second branch, while $|N(v)|=3$, leads to a correct solution. Let $C$ be the connected component that contains $v$. Since the previous rule did not apply, $|V(C)|>100$. As we continue making recursive calls, as long as $G$ contains edges from $E(C)$, it also
 contains at least one vertex of degree 1 or 2. For example, after removing $v$, it contains a neighbor of $v$ 
whose degree is 1, and after removing $N(v)$, it contains a neighbor of a vertex in $N(v)$ whose degree is 1 or 2.
 Now, before we remove all the edges in $E(C)$, we encounter a recursive call where $G$ contains a connected component of 
size at least 5, for which  Rule \ref{rule:deg2component} or \ref{rule:concomponent} is applicable.\footnote{The removal of $N(v)\cup\{v\}$ from $C$, which has maximum degree 3 and contains more than 100 vertices, 
generates at most 6 connected components; thus, it results in at least one component of at least $\lceil(101-4)/6\rceil=17$ vertices.}
Therefore, if we apply Rule \ref{rule:deg4} again (and even if we do not, but there is a solution as required), we first apply 
Rule \ref{rule:deg2component} or \ref{rule:concomponent} which decrease $k$ by 1, although the
 actual decrease is at least by 2. Indeed, 2 is the minimum size of any vertex cover of a connected component on at least 5 
vertices and of maximum degree 3. Thus, it is possible to decrease $k$ by 4 in the Rule  \ref{rule:deg4}. 
By the definition of this rule, its branching vector is at least as good as $(1,4)$, 
whose root is smaller than 1.381.}

\begin{reducerule}\label{rule:deg1}
{\normalfont [There are $v,u\in V$ such that $N(v)=\{u\}$]
\begin{enumerate}
\item If $w(v)\geq w(u)$: Return \alg{WVC-Alg}$(G[V\setminus \{v,u\}],w,W\!-w(u),k-1)\cup\{u\}$.
\smallskip
\item Else if there is $r\!\in\! V$ such that $N(u)\!=\!\{v,r\}$:
	\begin{enumerate}
	\smallskip
	\item Let $w'$ be $w$, except for $w'(r)=w(r)\!-\!(w(u)\!-\!w(v))$.
	\smallskip
	\item If $w'(r)\leq 0$: Return \alg{WVC-Alg}($G[V\!\setminus\! \{v,u,r\}],w',W\!-\!w(v)\!-\!w(r),k\!-\!1$), along with $\{v,r\}$.
	\smallskip
	\item Else: Return \alg{WVC-Alg}($G[V\!\setminus\! \{v,u\}],w',W\!-\!w(u),k\!-\!1$), along with $v$ if $r$ is in the returned result, and else along with $u$.
	\end{enumerate}
\smallskip
\item Else: Let $w'$ be $w$, except for $w'(u)=w(u)-w(v)$. Return \alg{WVC-Alg}($G[V\!\setminus \{v\}],w',W\!-w(v),k$), along with $v$ iff $u$ is not in the returned result.
\end{enumerate}}
\end{reducerule}

{\noindent This rule, illustrated below, omits leaves (i.e., if there is a leaf, $v$, it is omitted from $G$ in the recursive calls performed in this rule). Clearly, to obtain a solution, we should choose either $u$ or $N(u)$. If $w(v)\geq w(u)$ (Case 1), we simply choose $u$ (it is better to cover the only edge that touches $v$, $\{v,u\}$,~by~$u$).

Now, suppose that there is $r\in V$ such that $N(u)=\{v,r\}$. If $w'(r)\leq 0$ (Case 2b), it is better, in terms of weight, to choose $\{v,r\}$; yet, in terms of size, it might be better to choose $u$. In any case, $k$ should be decreased by at least 1. Our flexible use of the parameter $k$ allows us to decrease its value by 1, which is less than its actual decrease ($2=|\{v,u\}|$) in the instance.\footnote{In this manner, we may compute a vertex cover whose size is larger than $k$ (since we decrease $k$ 
only by 1), but we may not compute a vertex cover of weight larger than $W$. We note that if we decrease $k$ by 2, we
may overlook solutions: if there is a solution of size at most $k$ and weight at most $W$ that contains $u$, there is a solution of weight at most $W$ that contains $\{v,r\}$, but there {\em might not} be a solution of size at most $k$ and weight at most $W$ that contains $\{v,r\}$.} Next, suppose that $w'(r) > 0$ (Case 2c). In ($G[V\setminus \{v,u\}],w',W\!-w(u),k-1$), choosing $r$ reduces $W\!-w(u)$ to $W\!-w(u)-w'(r)=W\!-w(v)-w(r)$ and $k-1$ to $k-2$, which has the same effect as choosing $N(u)$ in the original instance. On the other hand, not choosing $r$ has the same effect as choosing $u$ in the original instance.

Finally, suppose that such $r$ does not exist (Case 3). In ($G[V\setminus \{v\}],w',W\!-w(v),k$), choosing $u$ reduces $W\!-w(v)$ to $W\!-w(v)-w'(u)=W\!-w(u)$ and $k$ to $k-1$, which has the same effect as choosing $u$ in the original instance. On the other hand, not choosing $u$ has {\em almost} the same effect as choosing $v$ in the original instance: the difference lies in the fact that we do not decrease $k$ by 1. However, our flexible use of the parameter $k$ allows us to decrease its value by less than necessary (as in Case 2b).\footnote{\alg{WVC-Alg} is not called with $k\!-\!1$, as then choosing $u$ overall decreases $k$ by 2, which is more than required (thus we may
 overlook solutions, by reaching Rule \ref{rule:vckneg} too soon).}}

\begin{figure}[!h]\centering
\frame{\includegraphics[scale=0.8]{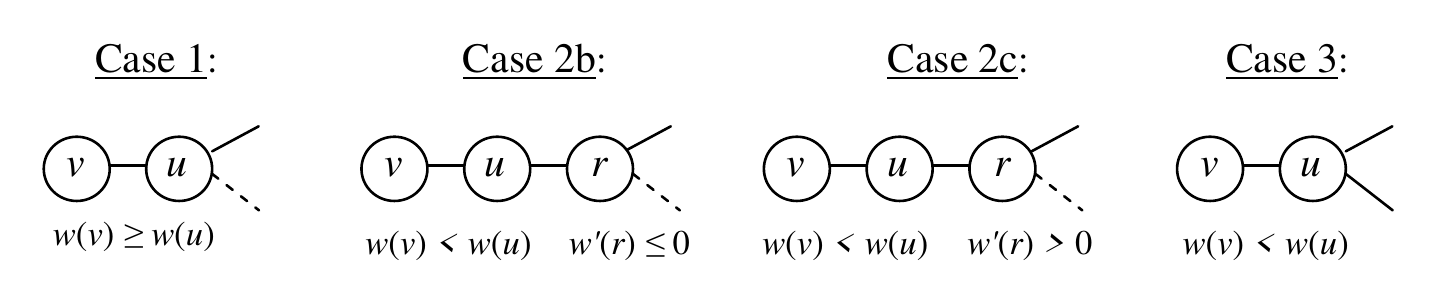}}
\caption{Rule \ref{rule:deg1} of \alg{WVC-Alg}.}
\label{fig:omitLeaves}
\end{figure}

\begin{reducerule}\label{rule:dominate}
{\normalfont [There are $v,u$ such that $v\in N(u)$, $N(u)\setminus \{v\} \subseteq N(v)\setminus \{u\}$ and $w(v)\leq w(u)$]
Return \alg{WVC-Alg}$(G[V\setminus \{v\}],w,W\!-w(v),k-1)\cup\{v\}$.}
\end{reducerule}

{\noindent The vertices $v$ and $u$ are neighbors; thus, we should choose at least one of them. If we do not choose $v$, we need to choose $N(v)$, in which case we can replace $u$ by $v$ and obtain a vertex cover (since $N(u)\!\setminus\! \{v\} \!\subseteq\! N(v)\!\setminus\! \{u\}$) of the same or better weight (since $w(v)\!\leq\! w(u)$). Thus, in this rule, we choose $v$. Note that, if there is a triangle with two degree-2 vertices, or exactly one degree-2 vertex that is heavier than one of the other vertices in the triangle, this rule deletes a vertex of the triangle. Thus, in the following rules, such triangles do not exist.}

\begin{reducerule}\label{rule:vctriangle}
{\normalfont  [There are $v,u,r$ such that $\{v,u\} = N(r)$, $\{v,r\} \subseteq N(u)$] Let $w'$ be $w$, except $w'(v)\!=\!w(v)\!-\!w(r)$ and $w'(u)\!=\!w(u)\!-\!w(r)$. Return \alg{WVC-Alg}($G[V\setminus \{r\}],w',W\!-\!2w(r),k$), along with $r$ iff not both $v$ and $u$ are in the returned~result.}
\end{reducerule}

{\noindent The rule is illustrated below. First, note that $w'(v),w'(u)\!>\! 0$ (otherwise Rule \ref{rule:dominate} applies), and thus calling \alg{WVC-Alg} with $w'$ is possible.

We need to choose {\em exactly} two vertices from $\{v,u,r\}$: choosing less than two vertices does not cover all three edges of the triangle, and choosing $r$, if $v$ and $u$ are already chosen, is unnecessary. In \alg{WVC-Alg}($G[V\setminus \{r\}],w',W-2w(r),k$), choosing only $v$ from $\{v,u\}$ reduces $W-2w(r)$ to $W-2w(r)-w'(v)=W-w(r)-w(v)$ and $k$ to $k-1$, which has {\em almost} the same effect as choosing $r$ and $v$ in the original instance, where the only difference lies in the fact that $k$ is reduced by 1 (rather than 2). However, our flexible use of the parameter $k$ allows us to decrease it by less than necessary. Symmetrically, choosing only $u$ from $\{v,u\}$ has almost the same effect as choosing $r$ and $u$ in the original instance, and again, our flexible use of the parameter $k$ allows us to decrease it by less than necessary. Finally, choosing both $v$ and $u$ reduces $W-2w(r)$ to $W-2w(r)-w'(v)-w'(u)=W-w(v)-w(u)$ and $k$ to $k-2$, which has the same effect as choosing $v$ and $u$ in the original instance. Thus, we have shown that each option of choosing exactly two vertices from $\{v,u,r\}$ is considered.}

\begin{figure}[!h]\centering
\frame{\includegraphics[scale=0.8]{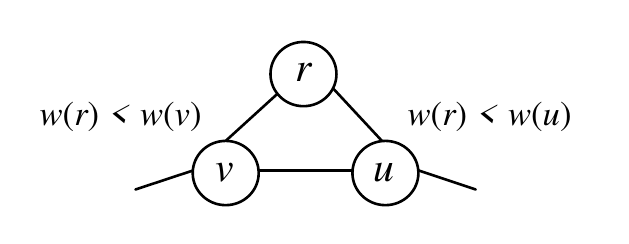}}
\caption{Rule \ref{rule:vctriangle} of \alg{WVC-Alg}.}
\label{fig:omitTriangles}
\end{figure}


{\noindent From now on, since previous rules did not apply, there are no connected components on at most 100 vertices (by Rule \ref{rule:concomponent}), no leaves (by Rule \ref{rule:deg1}), no vertices of degree at least 4 (by Rule \ref{rule:deg4}), and no triangles that contain a degree-2 vertex (by Rules \ref{rule:dominate} and \ref{rule:vctriangle}); also, there is a degree-2 vertex that is a neighbor of a degree-3 vertex (by Rules \ref{rule:deg2component} and \ref{rule:deg4}). We give the remaining rules in Appendix~\ref{app:wvc1}.}

\bibliographystyle{abbrv}

\newcounter{firstbib}

\appendix
\newpage

\section{Results Related to {\sc WVC}}\label{app:wvc}

In this appendix, we present algorithms and a hardness result related to {\sc WVC}. We first give (in Appendix \ref{app:wvc1}) the remaining rules of our $O^*(1.381^s)$ time and polynomial space algorithm (see Section \ref{section:wvc1}). Then, in Appendix \ref{section:wvc2}, we develop an $O^*(1.363^s)$ time and space algorithm. Appendix \ref{section:wvcbip} provides further evidence to the strength of our multivariate framework. Finally, we complement these results by developing (in Appendix \ref{section:wvc3}) an $O^*(1.443^t)$ time and polynomial space algorithm for {\sc WVC}, that is faster on graphs of bounded degree 3.

\subsection{An $O^*(1.381^s)$ Time Algorithm for {\sc WVC} (Cont.)}\label{app:wvc1}

We list the remaining rules used by \alg{WVC-Alg}. Note that each rule is followed by an illustration.

\begin{branchrule}\label{rule:triangles2}
{\normalfont [There are $v,u,r$ such that $\{v,u\} \!\subseteq\! N(r)$, $\{v,r\} \!\subseteq\! N(u)$ and $w(v)\!\leq\! w(u)\!\leq\! w(r)$] Let $\{u'\}=N(u)\!\setminus\! \{v,r\}$, $\{r'\}=N(r)\!\setminus\! \{v,u\}$, and $X \!=\! N(v)\!\cup\! N(u')\!\cup\! N(r')$.
\begin{enumerate}
\item If the result of \alg{WVC-Alg}($G[V\setminus \{v\}],w,W\!-w(v),k-1$) is not NIL: Return it along with~$v$.
\item Else: Return \alg{WVC-Alg}($G[V\setminus X],w,W-w(X),k-|X|)\cup X$.
\end{enumerate}}
\end{branchrule}

{\noindent If there is still a triangle in the graph, this rule omits at least one of its vertices. Therefore, after this rule, we may assume that $G$ does not contain triangles. Recall that, we have argued that at this point (i.e., after Rule \ref{rule:vctriangle}), a triangle contains only degree-3 vertices. If at least two of them had a common neighbor outside the triangle, Rule \ref{rule:dominate} would have been applied, and thus we may assume that each vertex in $\{v,u,r\}$ has a unique neighbor (which is the situation illustrated in Fig.~\ref{fig:omitTriangles2}). We should choose $v$ or $N(v)$. Choosing $N(v)$ (in the second branch), we do not choose $u'$ since then we can replace $u$ by $v$ without adding weight (and the choice of $v$ is already examined in the first branch); similarly, we do not choose $r'$ since then we can replace $r$ by $v$ without adding weight. Thus, choosing $N(v)$, we also choose the neighbors of $u'$ and $r'$ (i.e., we choose $X$).
Clearly, $|X|\geq|N(v)|\geq 3$. If $|X|=3$, then $N(u'),N(r')\subseteq N(v)$, which implies that $N(u')=N(v)\setminus\{r\}$ and $N(r')=N(v)\setminus\{u\}$, which is a contradiction (since then the triangle is part of a connected component of 6 vertices, which invokes Rule \ref{rule:concomponent}). Therefore, $|X|\geq 4$, and we get a branching vector that is at least as good as $(1,4)$, whose root is smaller than 1.381.}

\begin{figure}[!h]\centering
\frame{\includegraphics[scale=0.8]{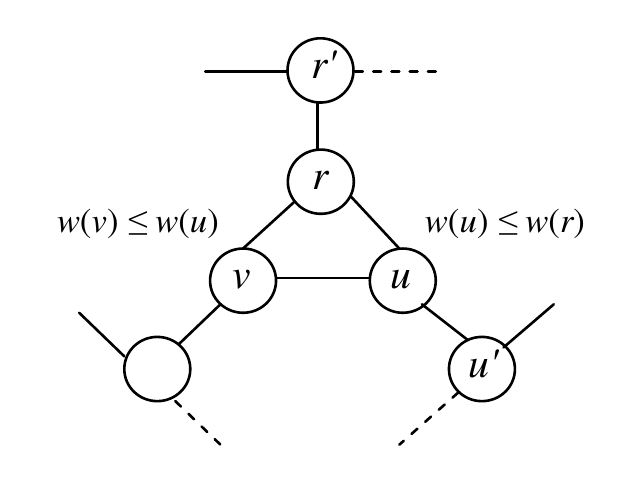}}
\caption{Rule \ref{rule:triangles2} of \alg{WVC-Alg}.}\label{fig:omitTriangles2}
\end{figure}

\begin{branchrule}\label{rule:deg2deg2}
{\normalfont [There are $v,u,r$ such that $|N(v)|=3$, $N(u)=\{v,r\}$ and $|N(r)|=2$]
\begin{enumerate}
\item If the result of \alg{WVC-Alg}($G[V\setminus \{v\}],w,W\!-w(v),k-1$) is not NIL: Return it along with~$v$.
\item Else: Return \alg{WVC-Alg}($G[V\setminus N(v)],w,W\!-w(N(v)),k-3)\cup N(v)$.
\end{enumerate}}
\end{branchrule}

{\noindent If there is a degree-2 vertex that is a neighbor of a degree-2 vertex, then there are $v,u$ and $r$ as defined in this rule. Therefore, after this rule, we may assume that the neighbors of a degree-2 vertex are degree-3 vertices. Clearly, the branching is exhaustive (we choose either $v$ or its neighbors). After choosing $v$, we apply Rule \ref{rule:deg1} (or a different rule that is at least as good)\footnote{When referring to a rule that is ``at least as good'' or ``better'', we refer only to preceding reduction rules {\em where $k$ is decreased by 1}.}\label{foot:better} where we decrease $k$ by 1 (i.e., we apply Case 1, 2b or 2c). Thus, we get a branching vector that is at least as good as $(1+1,3)=(2,3)$, whose root is smaller than 1.381.}

\begin{figure}[!h]\centering
\frame{\includegraphics[scale=0.8]{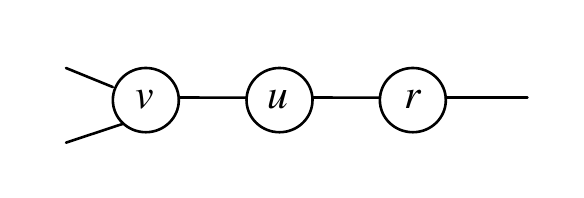}}
\caption{Rule \ref{rule:deg2deg2} of \alg{WVC-Alg}.}
\end{figure}

\begin{branchrule}\label{rule:vc11}
{\normalfont [There are $x,a,b,c,v$ such that $N(x)=\{a,b,c\}, N(a)=\{x,v\}, |N(b)|=2$, $|N(v)|=3$ and ($v\notin N(b)$ or $|N(c)|=2$)]
\begin{enumerate}
\item If the result of \alg{WVC-Alg}($G[V\setminus \{v\}],w,W\!-w(v),k-1$) is not NIL: Return it along with~$v$.
\item Else: Return \alg{WVC-Alg}($G[V\setminus N(v)],w,W-w(N(v)),k-3)\cup N(v)$.
\end{enumerate}}
\end{branchrule}

{\noindent This branching is exhaustive. After choosing $v$, we apply either Rule \ref{rule:deg1}, omitting $a$ (or $b$, if $v\in N(b)$), and then [Rule \ref{rule:deg2deg2} or a better rule (recall that the phrase ``a better rule'' was explained in the previous rule)], or a better rule. This is easily seen by noting that after omitting $a$ ($b$), either $x$ and $b$ are adjacent degree-2 vertices, or $b$ (resp.~$a$) is a leaf adjacent to the degree-2 vertex $x$. Furthermore, after choosing $N(v)$, we apply either Rule \ref{rule:deg2deg2}, or a better rule.
This is easily seen by noting that if $v\notin N(b)$, then after deleting $N(v)$, $x$ and $b$ are adjacent vertices of degree at most 2; else, $|N(c)|=2$ and thus $v\notin N(c)$ (else $v,a,x,b$ and $c$ form a small connected component which invokes Rule \ref{rule:concomponent}), and then after deleting $N(v)$, $x$ and $c$ are adjacent vertices of degree at most 2.
Thus, we get a branching vector that is at least as good as $(1+(2,3),3+(2,3))=(3,4,5,6)$, whose root is smaller than 1.381.}

\begin{figure}[!h]\centering
\frame{\includegraphics[scale=0.8]{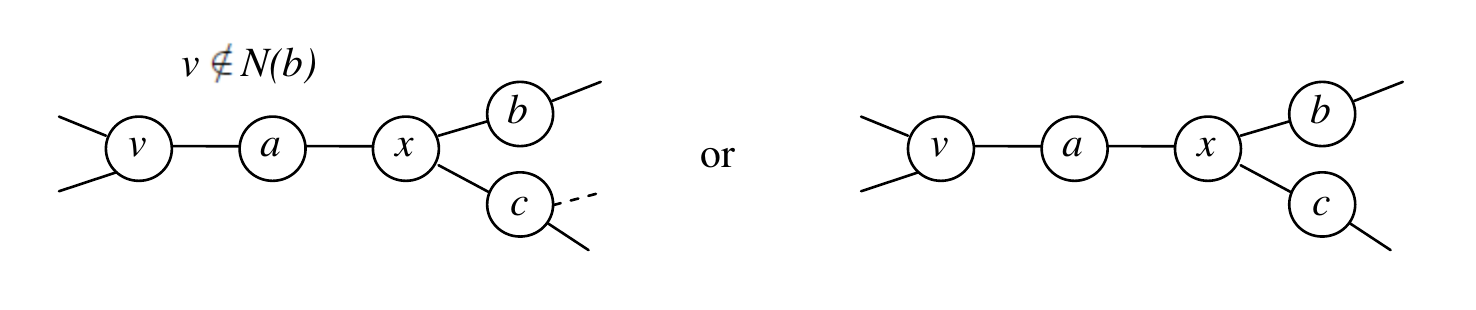}}
\caption{Rule \ref{rule:vc11} of \alg{WVC-Alg}. Note that the illustrated situations are not disjoint (it is possible that both $v\notin N(b)$ and $|N(c)|=2$).}
\end{figure}

\begin{branchrule}\label{rule:vc12}
{\normalfont [There are $x,a,b,c,v$ such that $N(x)=\{a,b,c\}, N(a)=\{x,v\}, N(b)=\{x,v\}$, $|N(v)|=3$ and $|N(c)|=3$]
\begin{enumerate}
\item If the result of \alg{WVC-Alg}($G[V\setminus \{c\}],w,W\!-w(c),k-1$) is not NIL: Return it along with~$c$.
\item Else: Return \alg{WVC-Alg}($G[V\setminus N(c)],w,W\!-w(N(c)),k-3)\cup N(c)$.
\end{enumerate}}
\end{branchrule}

{\noindent This branching is exhaustive. After choosing either $c$ or $N(c)$, we apply a combination of reduction rules that decreases $k$ by at least 1. More precisely, if $v\in N(c)$, then after choosing $c$, we apply Rule \ref{rule:deg2component}; else, after choosing $N(c)$, we apply, at worst, Case 3 in Rule \ref{rule:deg1} to omit $b$ (or $a$), and then we have a leaf that is adjacent to a degree-2 vertex, and can thus apply one of the other cases in Rule \ref{rule:deg1}. Thus, we get a branching vector that is at least as good as $\max\{(1+1,3),(1,3+1)\}=(1,4)$, whose root is smaller than 1.381.}

\begin{figure}[!h]\centering
\frame{\includegraphics[scale=0.8]{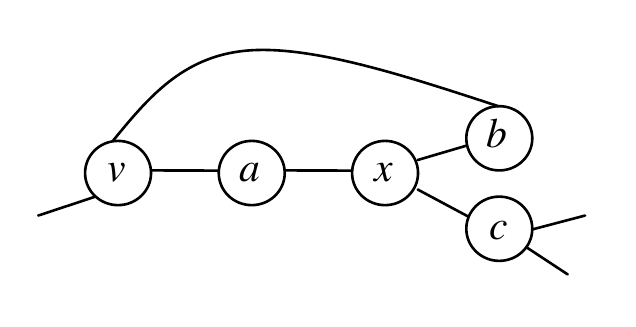}}
\caption{Rule \ref{rule:vc12} of \alg{WVC-Alg}.}
\end{figure}

{\noindent From now on, since previous rules did not apply, there are different vertices $x,a,b,c,v$ such that $N(x)=\{a,b,c\}$, $N(a)=\{x,v\}$, $|N(b)|=|N(c)|=|N(v)|=3$ (in particular, if $|N(v)|=2$, Rule \ref{rule:deg2deg2} is applied). Moreover, denoting $N(v)=\{a,v_1,v_2\}$, we have that $|N(v_1)|=|N(v_2)|=3$ (if the degree of at least one vertex in $\{b,c,v_1,v_2\}$ is 2, at worst, Rule \ref{rule:vc11} or \ref{rule:vc12} is applied). Note that neither $b$ and $c$, nor $v_1$ and $v_2$ are neighbors (since $G$ does not contain triangles). This situation is illustrated in the following figure.}

\begin{figure}[!h]\centering\label{fig:vc12after}
\frame{\includegraphics[scale=0.8]{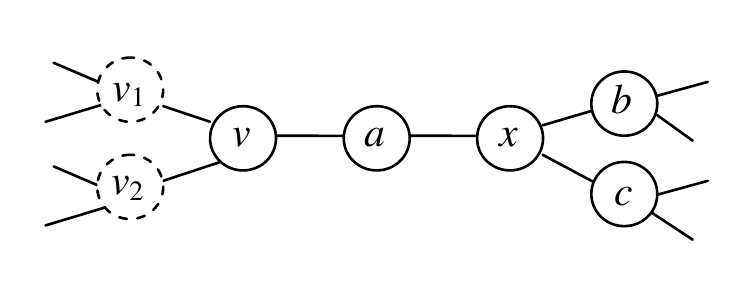}}
\caption{After Rule \ref{rule:vc12} of \alg{WVC-Alg}, the graph $G$ contains the illustrated subgraph. The vertices $v_1$ and $v_2$ appear in dashed circles since they may be equal to $b$ and $c$. Recall that $b\notin N(c)$ and $v_1\notin N(v_2)$.}
\end{figure}

\begin{branchrule}\label{rule:vc13}
{\normalfont [$b=v_1$ and $c=v_2$]
\begin{enumerate}
\item If the result of \alg{WVC-Alg}($G[V\!\setminus\! \{x,v\}],w,W\!-\!w(\{x,v\}),k\!-\!2$) is not NIL: Return it along with $\{x,v\}$.
\item Else: Return \alg{WVC-Alg}($G[V\setminus N(x)],w,W\!-w(N(x)),k-3)\cup N(x)$.
\end{enumerate}}
\end{branchrule}

{\noindent Clearly, we need to choose either $x$ or $N(x)$. Thus, to show that the rule is correct, it is enough to explain why, if we choose $x$ (in the first branch), we can also choose $v$. This can be easily seen by noting that if we choose $x$ and do not choose $v$, we need to choose $N(v)=N(x)$, in which case it is unnecessary to choose $x$. Note that we choose $N(x)$ in the second branch. We get a branching vector that is at least as good as $(2,3)$, whose root is smaller than 1.381.}

\begin{figure}[!h]\centering
\frame{\includegraphics[scale=0.8]{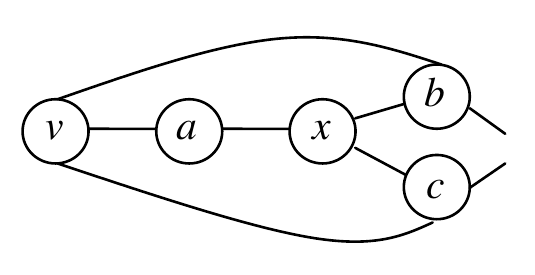}}
\caption{Rule \ref{rule:vc13} of \alg{WVC-Alg}.}
\end{figure}

\begin{branchrule}\label{rule:vc14}
{\normalfont [$b=v_1$ and $c\neq v_2$]
\begin{enumerate}
\item If the result of \alg{WVC-Alg}($G[V\setminus N(b)],w,W\!-w(N(b)),k-3$) is not NIL: Return it along with $N(b)$.
\item Else if the result of \alg{WVC-Alg}($G[V\setminus (\{b\}\cup N(c))],w,W\!-w(\{b\}\cup N(c)),k-4$) is not NIL: Return it along with $\{b\}\cup N(c)$.
\item Else: Return \alg{WVC-Alg}($G[V\setminus \{b,c\}],w,W\!-w(\{b,c\}),k-2)\cup\{b,c\}$.
\end{enumerate}}
\end{branchrule}

{\noindent This branching is exhaustive: we either choose $N(b)$ (branch 1) or $b$ (branches 2 and 3), where, if we choose $b$, we further consider choosing $N(c)$ (branch 2) or $c$ (branch 3). After choosing $\{b\}\cup N(c)$ (in the second branch), we apply a reduction rule that decreases $k$ at least by $1$ (Case 1 or 2 of Rule \ref{rule:deg1} to omit $a$, or a better rule). Moreover, after choosing $\{b,c\}$, we also apply a reduction rule that decreases $k$ by at least 1 (now $x$ is a leaf that is adjacent to a degree-2 vertex). Thus, we get a branching vector that is at least as good as $(3,4+1,2+1)=(3,5,3)$, whose root is smaller than 1.381.}

\begin{figure}[!h]\centering
\frame{\includegraphics[scale=0.8]{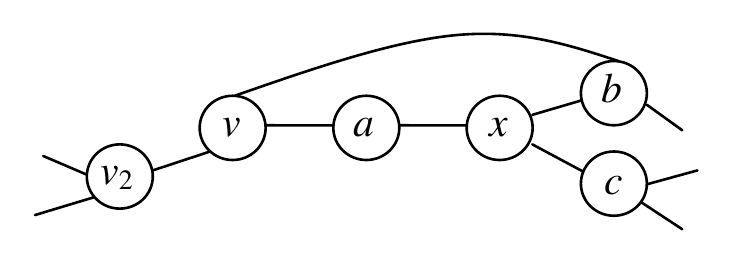}}
\caption{Rule \ref{rule:vc14} of \alg{WVC-Alg}.}
\end{figure}

{\noindent From now on, since the two previous rules did not apply, we may assume w.l.o.g that the vertices $x,a,b,c,v,v_1,v_2$ are different.}

\begin{branchrule}\label{rule:vc15}
{\normalfont [$b\in N(v_1)$]
\begin{enumerate}
\item If the result of \alg{WVC-Alg}($G[V\setminus \{b\}],w,W\!-w(b),k-1$) is not NIL: Return it along with~$b$.
\item Else: Return \alg{WVC-Alg}($G[V\setminus N(b)],w,W\!-w(N(b)),k-3)\cup N(b)$.
\end{enumerate}}
\end{branchrule}

{\noindent This branching is exhaustive. After choosing $N(b)$, we apply a reduction rule that decreases $k$ by at least 1 (Case 1 or 2 of Rule \ref{rule:deg1} to omit $a$, or a better rule). Thus, we get a branching vector that is at least as good as $(1,3+1)=(1,4)$, whose root is smaller than 1.381.}

\begin{figure}[!h]\centering
\frame{\includegraphics[scale=0.8]{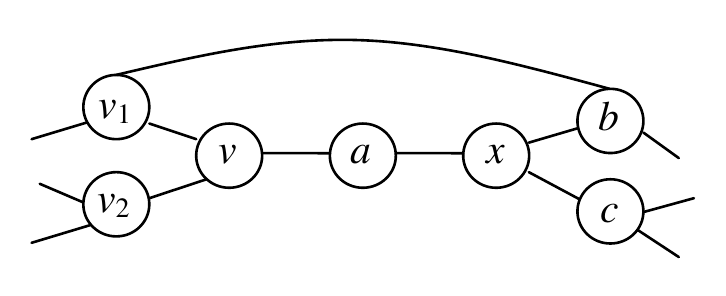}}
\caption{Rule \ref{rule:vc15} of \alg{WVC-Alg}. Note that, along with $\{b,v_1\}$, $\{b,c,v_1,v_2\}$ may contain other pairs of neighbors.}
\end{figure}

{\noindent From now on, since the previous rule did not apply, we may assume w.l.o.g that $\{b,c,v_1,v_2\}$ does not contain vertices that are neighbors.}

\begin{branchrule}\label{rule:vc16}
{\normalfont [Remaining case]
\begin{enumerate}
\item If the result of \alg{WVC-Alg}($G[V\setminus N(b)],w,W\!-w(N(b)),k-3$) is not NIL: Return it along with $N(b)$.
\item Else if the result of \alg{WVC-Alg}($G[V\setminus (\{b\}\cup N(v_2))],w,W\!-w(\{b\}\cup N(v_2)),k-4$) is not NIL: Return it along with $\{b\}\cup N(v_2)$.
\item Else if the result of \alg{WVC-Alg}($G[V\setminus (\{b,v_2\}\cup N(c))],w,W\!-w(\{b,v_2\}\cup N(c)),k-5$) is not NIL: Return it along with $\{b,v_2\}\cup N(c)$.
\item Else: Return \alg{WVC-Alg}($G[V\setminus \{b,v_2,c\}],w,W\!-w(\{b,v_2,c\}),k-3)\cup \{b,v_2,c\}$.
\end{enumerate}}
\end{branchrule}

{\noindent This branching is exhaustive. It is easy to verify that in each branch, except for the first, we apply a reduction rule that decreases $k$ by at least 1. Indeed, after choosing $\{b\}\cup N(v_2)$, we apply Case 1 or 2 of Rule \ref{rule:deg1} to omit $a$, or a better rule; after choosing $\{b,v_2\}\cup N(c)$, we also apply Case 1 or 2 of Rule \ref{rule:deg1} to omit $a$, or a better rule (note that, since $b,v_2\notin N(c)$, $|\{b,v_2\}\cup N(c)|=5$); finally, after choosing $\{b,v_2,c\}$, we apply Case 1 or 2 of Rule \ref{rule:deg1} to omit $x$, or a better rule. Thus, we get a branching vector that is at least as good as $(3,4+1,5+1,3+1)=(3,5,6,4)$, whose root is smaller than 1.381.}

\begin{figure}[!h]\centering
\frame{\includegraphics[scale=0.8]{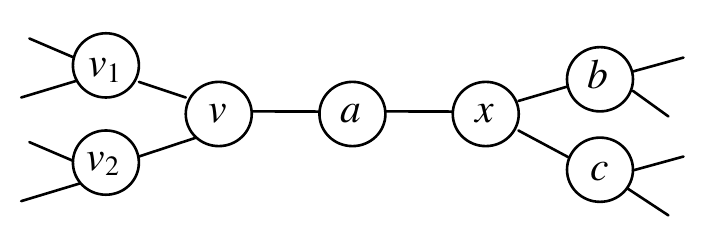}}
\caption{Rule \ref{rule:vc16} of \alg{WVC-Alg}. Recall that $\{b,c,v_1,v_2\}$ does not contain neighbors.}
\end{figure}
\subsection{A Faster Exponential-Space Algorithm for {\sc WVC}}\label{section:wvc2}

In this appendix, we solve {\sc WVC} in time and space $O^*(1.363^s)$. For this purpose, we consider the following form of {\sc $k$-WVC}.

\myparagraph{$k$-WVCnoW} Given an instance of {\sc $k$-WVC}, return a vertex cover $U$ whose weight is smaller or equal to the weight of any vertex cover of size at most $k$.

\medskip

Clearly, when discussing an instance of {\sc $k$-WVCnoW}, we can omit the weight $W$. This is useful due to the fact that at a later stage in the appendix, we store ``subinstances'' of the input instance, along with their solutions, and would like to avoid storing the values of $W$ that correspond to these subinstances.

Next, by modifying \alg{WVC-Alg} (the algorithm given in Section \ref{section:wvc1}), we develop a bounded search tree-based algorithm, called \alg{WVCnoW-Alg}, for which we obtain the following result.

\begin{lem}\label{lemma:wvc2}
\alg{WVCnoW-Alg} solves {\sc $k$-WVCnoW} in time $O^*(1.3954^k)$ and polynomial space. Given an input $(G\!=\!(V\!,E),w\!: V\!\!\rightarrow\! \mathbb{R}^{\geq 1}\!,k)$, any node in the search tree corresponds to an instance $(G'\!=\!(V'\!,E'),w'\!: V'\!\!\rightarrow\! \mathbb{R}^{\geq 1}\!,k')$ where $G'$ is an induced subgraph of $G$, $w'$ is defined as $w$ when restricted to $V'$,\footnote{This condition is not satisfied by \alg{WVC-Alg}, since it changes the weight of certain vertices in Rules \ref{rule:deg1} and \ref{rule:vctriangle} in Section \ref{section:wvc1}. Thus, we now need to replace Rules \ref{rule:deg1} and \ref{rule:vctriangle} by different rules, and also modify some of the rules that relied on them (i.e., the rules whose application was argued to be followed by an application of Rule \ref{rule:deg1} or \ref{rule:vctriangle}).} and $k'\!\leq\! k$. Moreover, \alg{WVCnoW-Alg} uses branching rules whose vectors have roots that are smaller than 1.3954, and stops if $k'\!<\!0$.
\end{lem}

\begin{proof}
We now present each rule related to a call \alg{WVCnoW-Alg}($G=(V,E),w: V\rightarrow \mathbb{R}^{\geq 1},k$), argue its correctness, and, if it is a branching rule, give its root (with respect to $k$). Since the worst root we shall get is bounded by 1.3954, and the algorithm stops if $k<0$, we get the desired running time.

When presenting a branching rule, we write ``Return ...'' at each branch. This indicates that we perform all the recursive calls related to the branching rule, and return the lightest vertex cover (i.e., the one whose total weight, according to $w$, is the smallest) among those returned by the branches (e.g., the branches of Rule \ref{rule:deg1a} below return two solutions, and we return the lighter one among them). If there are several vertex covers that have the same lightest weight, choose one of them arbitrarily.

\bigskip

{\noindent\bf Rules 1-5.} Use Rules 1-5 of \alg{WVC-Alg}, ignoring the weight $W$, and returning the lightest solution when branching (as explained above).

\smallskip

\setcounter{reducerule}{5}

\begin{branchrule}\label{rule:deg1a}
{\normalfont [There are different $v,u,r\in V$ such that $N(v)=\{u\}$, ($|N(u)|=2$ or $N(u)$ includes two leaves), $|N(r)|=3$, and there is a path from $u$ to $r$ consisting only of degree-2 internal vertices\footnote{If $u\in N(r)$, the last requirement in the condition is satisfied as there is a path from $u$ to $r$ that does not contain internal vertices.}]
\begin{enumerate}
\item Return \alg{WVCnoW-Alg}($G[V\setminus \{r\}],w,k-1$) $\cup\ \{r\}$.
\item Return \alg{WVCnoW-Alg}($G[V\setminus N(r)],w,k-3$) $\cup\ N(r)$.
\end{enumerate}}
\end{branchrule}

{\noindent This branching is exhaustive. After choosing $r$, we get that $v$ and $u$ are contained in a connected component that has at most one degree-3 vertex (which is, possibly, $u$); thus, we decrease $k$ by at least 1 by applying Rule \ref{rule:deg2component}. Thus, we get a branching vector that is at least as good as $(1+1,3)=(2,3)$, whose root is smaller than 1.3954.}

\begin{figure}[!h]\centering
\frame{\includegraphics[scale=0.8]{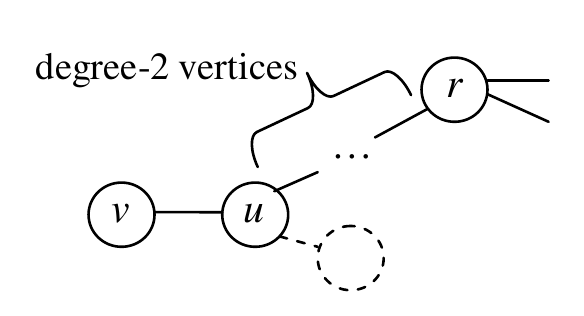}}
\caption{Rule \ref{rule:deg1a} of \alg{WVCnoW-Alg}.}
\end{figure}

\begin{branchrule}\label{rule:deg1b}
{\normalfont [There are $v,u,r\in V$ such that $N(v)=\{u\}$, $|N(r)|=3$, and there is a path from $u$ to $r$ consisting only of degree-2 internal vertices]
\begin{enumerate}
\item Return \alg{WVCnoW-Alg}($G[V\setminus \{r\}],w,k-1$) $\cup\ \{r\}$.
\item Return \alg{WVCnoW-Alg}($G[V\setminus N(r)],w,k-3$) $\cup\ N(r)$.
\end{enumerate}}
\end{branchrule}

{\noindent This branching is exhaustive. Note that, since the previous rule did not apply, the neighbors of $u$ include the leaf $v$, and two other vertices of degree at least 2. Thus, after choosing $r$, we clearly apply Rule \ref{rule:deg1a} or a better rule (recall that the term ``better rule'' was defined in Appendix \ref{app:wvc1}). Thus, we get a branching vector that is at least as good as $(1+(2,3),3)=(3,4,3)$, whose root is smaller than 1.3954.}

\begin{figure}[!h]\centering
\frame{\includegraphics[scale=0.8]{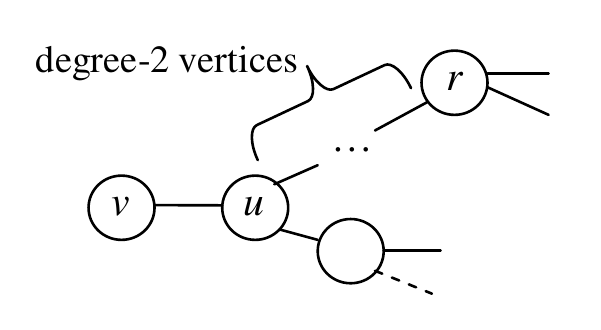}}
\caption{Rule \ref{rule:deg1b} of \alg{WVCnoW-Alg}.}
\end{figure}

\bigskip

{\noindent From now on, since Rules \ref{rule:deg2component}, \ref{rule:deg1a} and \ref{rule:deg1b} did not apply, we may assume that the graph $G$ does not contain leaves.}

\bigskip

{\noindent\bf Reduction Rules 8.} Use Rule 7 of \alg{WVC-Alg}, ignoring the weight $W$.

\smallskip

\setcounter{reducerule}{8}

\begin{branchrule}\label{rule:2vc9}
{\normalfont [There are $v\!,u,\!r$ such that $\{v,\!u\} \!=\! N(r)$ and $\{v,\!r\} \!\subseteq\! N(u)$]
\begin{enumerate}
\item Return \alg{WVCnoW-Alg}($G[V\setminus \{v\}],w,k-1$) $\cup\ \{v\}$.
\item Return \alg{WVCnoW-Alg}($G[V\setminus N(v)],w,k-3$) $\cup\ N(v)$.
\end{enumerate}}
\end{branchrule}

{\noindent This rule is illustrated in Fig.~\ref{fig:omitTriangles} (the rule to which it refers is related to the same condition as this rule). The branching is exhaustive. After choosing $v$, we clearly apply Rule \ref{rule:deg1a} or a better rule. Thus, we get a branching vector that is at least as good as $(1+(2,3),3)=(3,4,3)$, whose root is smaller than 1.3954.}

\bigskip

{\noindent\bf Reduction Rules 10-11.} Use Rules 9-10 of \alg{WVC-Alg}, ignoring the weight $W$, and returning the lightest solution.

\bigskip

{\noindent Note that in Rule 11 (that is Rule 10 of \alg{WVC-Alg}), since Rule \ref{rule:deg1a} is worse than Rule \ref{rule:deg1} of \alg{WVC-Alg}, we now obtain a branching vector that is at least as good as $(1+(2,3),3)=(3,4,3)$ (but not as $(2,3)$), whose root is smaller than 1.3954. }

\bigskip

{\noindent From now on, since previous rules did no apply, there are no connected components on at most 100 vertices (by Rule \ref{rule:concomponent}), no vertices of degree at least 4 (by Rule \ref{rule:deg4}), no leaves (by Rules \ref{rule:deg1a} and \ref{rule:deg1b}), no triangles (by Rules 8, \ref{rule:2vc9} and 10), and no degree-2 vertices that are neighbors (by Rule 11); also, there is a degree-2 vertex that is a neighbor of a degree-3 vertex (by Rules \ref{rule:deg2component} and \ref{rule:deg4}). In particular, this implies that there are different vertices $x,a,b,c,v$ such that $N(x)=\{a,b,c\}$, $N(a)=\{x,v\}$ and $|N(v)|=3$.}

\setcounter{reducerule}{11}

\begin{branchrule}\label{rule:2vc12}
{\normalfont [$N(b)=\{x,v\}$\footnote{Note that this condition includes the case where $N(c)=\{x,v\}$ (rename $b$ as $c$ and vice versa).}]
\begin{enumerate}
\item Return \alg{WVCnoW-Alg}($G[V\setminus \{v\}],w,k-1$) $\cup\ \{v\}$.
\item Return \alg{WVCnoW-Alg}($G[V\setminus N(v)],w,k-3$) $\cup\ N(v)$.
\end{enumerate}}
\end{branchrule}

{\noindent This branching is exhaustive. After choosing $v$, we apply Rule \ref{rule:deg1a} or a better rule. Thus, we get a branching vector that is at least as good as $(1+(2,3),3)=(3,4,3)$, whose root is smaller than 1.3954.}

\begin{figure}[!h]\centering
\frame{\includegraphics[scale=0.8]{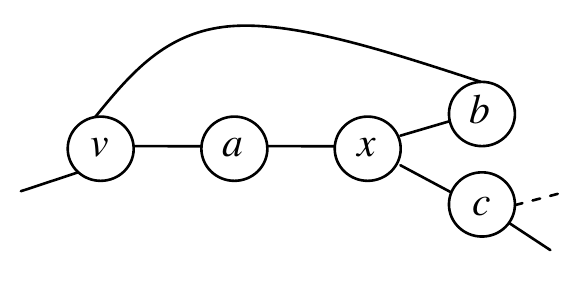}}
\caption{Rule \ref{rule:2vc12} of \alg{WVC-Alg}.}
\end{figure}

\begin{branchrule}\label{rule:2vc13}
{\normalfont [($|N(b)|=|N(c)|=2$) and $w(x)\geq w(b)+w(c)$]
\begin{enumerate}
\item Return \alg{WVCnoW-Alg}($G[V\setminus \{v\}],w,k-1$) $\cup\ \{v\}$.
\item Return \alg{WVCnoW-Alg}($G[V\setminus (N(v)\cup\{b,c\})],w,k-4$) $\cup\ N(v)\cup\{b,c\}$.
\end{enumerate}}
\end{branchrule}

{\noindent Since the previous rule was not applied, $v\notin N(b),N(c)$. In this rule, we choose either $v$ or $N(v)$. Choosing $N(v)$, we should further choose $x$ or $\{b,c\}=N(x)\setminus N(v)$. Since $w(x)\geq w(b)+w(c)$, it is always better, in terms of weight, to choose $\{b,c\}$. However, in terms of size, it might be better to choose $x$. Thus, we choose $\{b,c\}$, but reduce $k$ only by $4$ (even though we choose five vertices). Recall that reducing $k$ by less than its actual decrease in the instance complies with our flexible use of $k$ (see, e.g., Rule \ref{rule:deg1} in Section \ref{section:wvc1}). We get a branching vector that is at least as good as $(1,4)$, whose root is smaller than 1.3954.}

\begin{figure}[!h]\centering
\frame{\includegraphics[scale=0.8]{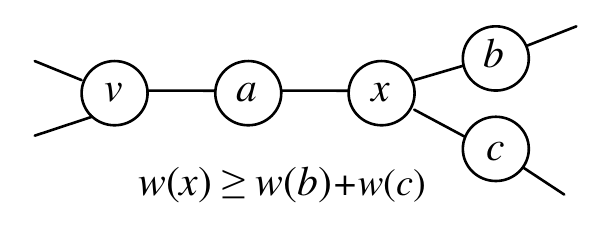}}
\caption{Rule \ref{rule:2vc13} of \alg{WVC-Alg}. Recall that $v\notin N(b),N(c)$.}
\end{figure}

\begin{branchrule}\label{rule:2vc14}
{\normalfont [There are $b',c'\in V$ such that $N(b)=\{x,b'\}$, $N(c)=\{x,c'\}$ and $b'\notin N(c')$]
\begin{enumerate}
\item Return \alg{WVCnoW-Alg}($G[V\setminus N(b')],w,k-3$) $\cup\ N(b')$.
\item Return \alg{WVCnoW-Alg}($G[V\setminus (\{b'\}\cup N(c'))],w,k-4$) $\cup\ \{b'\}\cup N(c')$.
\item Return \alg{WVCnoW-Alg}($G[V\setminus \{b',c',x\}],w,k-3$) $\cup\ \{b',c',x\}$.
\end{enumerate}}
\end{branchrule}

{\noindent Note that $|N(b')=|N(c')|=3$ (since Rule 10 did not apply), $b',c'\neq v$ and $b'\neq c'$ (since Rule \ref{rule:2vc12} did not apply; the latter argument can be easily seen by reshuffling the names of the vertices---in particular, $b'$ is $v$ in Rule \ref{rule:2vc12}), and $w(x)< w(b)+w(c)$ (since Rule \ref{rule:2vc13} did not apply). In this rule, we choose either $N(b')$ (branch 1) or $b'$ (branches 2 and 3), where if we choose $b'$, we further choose either $N(c')$ (branch 2; note that, by the condition of this rule, $|\{b'\}\cup N(c')|=4$) or $c'$ (branch 3). Choosing $\{b',c'\}$ (in branch 3), we can further choose $x$ (since $w(x)< w(b)+w(c)$). We get a branching vector that is at least as good as $(3,4,3)$, whose root is smaller than 1.3954.}

We do not next consider a rule that is similar to this rule, except that $b'\in  N(c')$, since there are at least two vertices among $\{v,b',c'\}$ that are not neighbors (otherwise we have a connected component of 7 vertices), which invokes Rule \ref{rule:2vc14} (this can be easily seen by reshuffling the names of the vertices). Thus, from now on, at least one vertex in $\{b,c\}$ has degree 3, and we may assume w.l.o.g that this vertex is $c$. Also note that, since Rule \ref{rule:2vc12}, if $|N(b)|=2$, then $v\notin N(b)$.

\begin{figure}[!h]\centering
\frame{\includegraphics[scale=0.8]{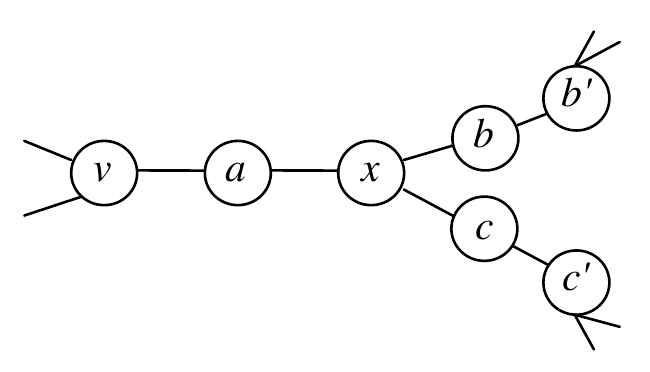}}
\caption{Rule \ref{rule:2vc14} of \alg{WVC-Alg}. Recall that $w(x)< w(b)+w(c)$.}
\end{figure}

\begin{branchrule}\label{rule:2vc15}
{\normalfont [$w(a) \geq \min\{w(x),w(v)\}$]
Let $p$ be a vertex in $\{x,v\}$ such that $w(a)\geq w(p)$, and $q$ be the other vertex in $\{x,v\}$.
\begin{enumerate}
\item Return \alg{WVCnoW-Alg}($G[V\setminus N(q)],w,k-3$) $\cup\ N(q)$.
\item Return \alg{WVCnoW-Alg}($G[V\setminus (\{q,p\})],w,k-2$) $\cup\ \{q,p\}$.
\end{enumerate}}
\end{branchrule}

{\noindent We choose either $N(q)$ or $q$. Choosing $q$, we can simply add $p$ (since $w(a)\geq w(p)$). We get a branching vector that is at least as good as $(3,2)$, whose root is smaller than 1.3954.}

\begin{figure}[!h]\centering
\frame{\includegraphics[scale=0.8]{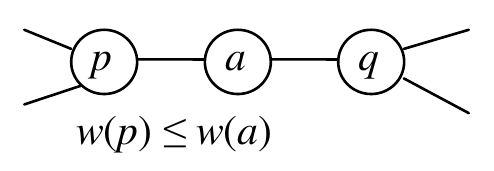}}
\caption{Rule \ref{rule:2vc15} of \alg{WVC-Alg}.}
\end{figure}

{\noindent From now on, since the previous rule did not apply, $w(a)\!<\!w(x)$ and $w(a)\!<\!w(v)$.}

\begin{branchrule}
{\normalfont [$|N(b)|=3$]
\begin{enumerate}
\item Return \alg{WVCnoW-Alg}($G[V\setminus N(b)],w,k-3$) $\cup\ N(b)$.
\item Return \alg{WVCnoW-Alg}($G[V\setminus (\{b\}\cup N(c))],w,k-4$) $\cup\ \{b\}\cup N(c)$.
\item Return \alg{WVCnoW-Alg}($G[V\setminus \{b,c,a\}],w,k-3$) $\cup\ \{b,c,a\}$.
\end{enumerate}}
\end{branchrule}

{\noindent We either choose $N(b)$ (branch 1) or $b$ (branches 2 and 3), where if we choose $b$, we further choose $N(c)$ (branch 2) or $c$ (branch 3). Choosing $\{b,c\}$, we can simply add $a$ (since $w(a)<w(x)$). We get a branching vector that is at least as good as $(3,4,3)$, whose root is smaller than 1.3954.}

\bigskip

{\noindent Next, let $N(v)=\{a,v_1,v_2\}$, and assume w.l.o.g that $|N(v_1)|=|N(b)|=2$. Also, let $N(b)=\{x,b'\}$. Recall that $b'\notin\{c,v\}$, and further assume w.l.o.g that $w(b)<\min\{w(x),w(b')\}$ (this can be assumed since Rule \ref{rule:2vc15} did not apply).}

\begin{branchrule}\label{rule:2vc17}
{\normalfont [Remaining case]
\begin{enumerate}
\item Return \alg{WVCnoW-Alg}($G[V\setminus N(c)],w,k-3$) $\cup\ N(c)$.
\item Return \alg{WVCnoW-Alg}($G[V\setminus N(x)],w,k-3$) $\cup\ N(x)$.
\item Return \alg{WVCnoW-Alg}($G[V\setminus \{c,x,v,b'\}],w,k-4$) $\cup\ \{c,x,v,b'\}$.
\end{enumerate}}
\end{branchrule}

{\noindent We either choose $N(c)$ (branch 1) or $c\in N(x)$ (branches 2 and 3), where if we choose $c$, we further choose $N(x)$ (branch 2) or $x$ (branch 3). Choosing $c$ and $x$, we do not need to try and choose $a$, since then $x$ can be replaced by $b$ without adding weight (since $w(b)\leq w(x)$), and the choice of $N(x)$ is already examined in the second branch; similarly, we do not need to choose $b$. Therefore, choosing $c$ and $x$, we also choose $\{v,b'\}\subseteq N(a)\cup N(b)$. We get a branching vector that is at least as good as $(3,3,4)$, whose root is smaller than 1.3954.}\qed

\begin{figure}[!h]\centering
\frame{\includegraphics[scale=0.8]{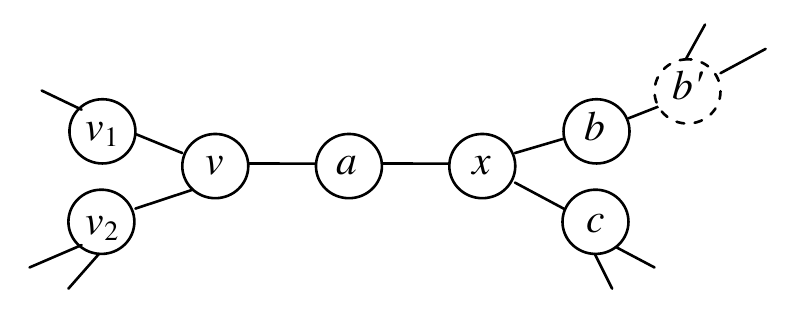}}
\caption{Rule \ref{rule:2vc17} of \alg{WVC-Alg}. The vertex $b'$ appears in a dashed circle since it is possible that $v_2=b'$.}
\end{figure}
\end{proof}


{\noindent Relying on \alg{WVCnoW-Alg}, we next apply the refined memorization technique of \cite{vc2005} (based on \cite{memorization1,memorization2}), as explained in this paragraph, to solve {\sc $k$-WVC} faster. Following this memorization technique, we store instances $(G',w',k')$ that correspond to nodes in the search tree, along with their solutions. For each node corresponding to a connected graph $G'$, we first look up in the stored instances: if the instance is found, we return its solution; otherwise, we continue running \alg{WVCnoW-Alg}, and when we return to this node, we store the instance along with its solution. Given a graph $G'$ that is not connected, we use an efficient branching rule, given in \cite{vc2005}, and show that its correctness, in our context, is preserved due to on our flexible use of the parameter $k$.}

To benefit from storing solutions for subproblems, nodes should correspond to instances where $G'$ contains at most $c\cdot k'$ vertices, for a small constant $c$.\footnote{If the nodes can correspond to large instances (which includes, in particular, instances where $G'$ contains many vertices), the solution storage can be too large, which does not result in improved running times.} For {\sc VC}, one can simply use the $2t$-vertex kernel, given in \cite{vc2001}, to reduce the size of $G'$. For {\sc WVC}, there is a $2W$-vertex kernel \cite{wvcker2008}, but no $c\cdot t$-vertex kernel exists \cite{wvcker2013}. However, we can obtain a $(\Delta\cdot k'+1)$-vertex kernel for any connected graph $G'$, where $\Delta$ is the maximum degree of a vertex in $G'$, as follows. By the Buss rule \cite{vc1993}, if $G$ contains more than $\Delta k'$ edges, there is no vertex cover of size at most $k'$ (since $k'$ vertices can be the endpoints of at most $\Delta k'$ edges), and thus we can simply return NIL. A connected graph having at most $\Delta k'$ edges has at most $\Delta k'+1$ vertices. Hence, we obtain a $(\Delta k'+1)$-vertex kernel.

Our vertex kernel slightly complicates the application of the refined memorization technique. We apply it only if $\Delta\in\{3,4\}$, since only then we have a kernel that is small enough to justify storing the solutions. While $\Delta\geq 5$, we apply a branching rule whose vector is at least as good as $(1,5)$, whose root is smaller than 1.325. Then, while $\Delta=4$, we apply a branching rule whose vector is at least as good as $(1,4)$, whose root is smaller than 1.3803. Combined with our solution-storage, this phase can be executed in time $O^*(1.363^{k'})$. Finally, while $\Delta=3$, we run \alg{WVCnoW-Alg}, that uses branching rules whose vectors have roots that are only smaller than 1.3954. However, when $\Delta=3$, we compute better vertex kernels. Thus, combined with our solution-storage, this phase can be executed in time $O^*(1.362^{k'})$. The technical details are given in the proof of the following theorem.

\begin{theorem}\label{theorem:wvc2}
{\sc $k$-WVC} can be solved in $O^*(1.363^k)$ time and space. Moreover, {\sc WVC} can be solved in $O^*(1.347^W)$ time and space.
\end{theorem}

By the discussion in Section \ref{section:technique}, this implies the desired result:

\begin{cor}
{\sc WVC} can be solved in $O^*(1.363^s)$ time and space.
\end{cor}

\begin{proof}[Theorem~\ref{theorem:wvc2}]
We start by considering the first, more interesting part of the theorem. We apply the refined memorization technique on \alg{WVCnoW-Alg}, as explained in our overview of this proof, to obtain an algorithm that we call \alg{WVCnoW-Alg2}. Clearly, it is enough to prove that \alg{WVCnoW-Alg2} solves {\sc $k$-WVCnoW} in $O^*(1.363^k)$ time and space, since to solve {\sc $k$-WVC}, we simply need to run \alg{WVCnoW-Alg2} and return its solution iff its weight is at most $W$.

Let ${\cal I}=(G=(V,E),w: V\rightarrow \mathbb{R}^{\geq 1},k)$ be the given instance of {\sc $k$-WVCnoW}, and let ${\cal S}$ be our solution-storage, which is initially empty. Given a node in the search tree, let ${\cal I}'=(G'=(V',E'),w': V'\rightarrow \mathbb{R}^{\geq 1},k')$ denote its corresponding instance. Let $\Delta'$ be the maximum degree of a vertex in $G'$. We next present our reduction and branching rules.

\setcounter{reducerule}{0}

\begin{reducerule}
{\normalfont[${\cal I}'\in {\cal S}$]
Return the solution associated with ${\cal I}'$ (that is stored in ${\cal S}$).}
\end{reducerule}

{\noindent We exploit our solution-storage; thus, we do not solve the same instance more than once.}

\begin{reducerule}\label{rule:kernel}
{\normalfont [$|E'|> \Delta'k'$]
Return NIL.}
\end{reducerule}

{\noindent If $|E'|> \Delta'k'$, there is no vertex cover of size at most $k'$, and thus we return NIL. Recall that we explained the intuition behind the necessity of this rule in the overview given in this section.}

\begin{branchrule}\label{rule:memdeg5}
{\normalfont [$\Delta' \geq 5$]
Let $v$ be a vertex of maximum degree.
\begin{enumerate}
\item Return \alg{WVCnoW-Alg2}($G[V'\setminus \{v\}],w',k'-1$) $\cup\ \{v\}$.
\item Return \alg{WVCnoW-Alg2}($G[V'\setminus N(v)],w',k'-|N(v)|$) $\cup\ N(v)$.\footnote{Recall that this standard form of presentation of a branching rule was used and explained for \alg{WVCnoW-Alg} (see the proof of Lemma \ref{lemma:wvc2})}
\end{enumerate}}
\end{branchrule}

{\noindent We perform a simple branching that handles vertices of degree at least 5. We get a branching vector that is at least as good as $(1,5)$, whose root is smaller than 1.325.}

\begin{reducerule}\label{rule:easycom}
{\normalfont [There is a connected component of at most one vertex of degree at least 3, or of at most 10 vertices]
Choose the first applicable rule of \alg{WVCnoW-Alg}.}
\end{reducerule}

\begin{branchrule}\label{rule:notcontrule}
{\normalfont [$G'$ is not a connected]
Let $H_1,H_2,\ldots,H_{\ell}$ denote the connected components of $G'$. Also, let $\widetilde{k}=k'-3(\ell-1)$. If $\widetilde{k}<3$, return NIL. Otherwise perform the following branching, returning the {\em union} of its solutions.
\begin{itemize}
\item Return \alg{WVCnoW-Alg2}($H_1,w',\widetilde{k}$).
\item Return \alg{WVCnoW-Alg2}($H_2,w',\widetilde{k}$).
\item $\ldots$
\item Return \alg{WVCnoW-Alg2}($H_{\ell},w',\widetilde{k}$).
\end{itemize}}
\end{branchrule}

{\noindent Since previous rules did not apply, the maximum degree of $H_i$, for any $1\leq i\leq \ell$, is 3 or 4 (by Rules \ref{rule:memdeg5} and \ref{rule:easycom}), and it contains at least 11 vertices (by Rule \ref{rule:easycom}), at thus at least 10 edges. To cover at least 10 edges in $H_i$, we need at least $\lceil 10/4\rceil=3$ vertices. Thus, any vertex cover of $H_i$ contains at least 3 vertices; thus, since any vertex cover of $G'$ contains a vertex cover of each $H_i$, we get that any vertex cover of size at most $k'$ of $G'$, including such vertex cover of minimum weight, contains a vertex cover of size at most $\widetilde{k}$ of $H_i$. Therefore, the branching performed in this rule is correct. In particular, note that is relies on our flexible use of the parameter $k$. For example, if $k'=20$, $\ell=5$ and thus $\widetilde{k}=8$, and the size of the vertex cover we compute for each $H_i$ is $7$, then the total size of the vertex cover we compute for $G'$ is $5\cdot 7 > k'$. We get the branching vector $(3(\ell-1),3(\ell-1),\ldots,3(\ell-1))$, where $3(\ell-1)$ appears $\ell$ times, whose root is smaller than 1.325 (indeed, for $\ell\geq 2$, $1.325 > \ell^{\frac{1}{3(\ell-1)}}$).}

\begin{branchrule}
{\normalfont [$\Delta' \geq 4$]
Let $v$ be a vertex of maximum degree.
\begin{enumerate}
\item Return \alg{WVCnoW-Alg2}($G[V'\setminus \{v\}],w',k'-1$) $\cup\ \{v\}$.
\item Return \alg{WVCnoW-Alg2}($G[V'\setminus N(v)],w',k'-4$) $\cup\ N(v)$.\end{enumerate}
Update the solution-storage.}
\end{branchrule}

{\noindent We perform a simple branching that handles degree-4 vertices. We get the branching vector $(1,4)$, whose root is smaller than 1.3803.}

\begin{reducebranchrule}
{\normalfont [Remaining case]
Choose the first applicable rule of \alg{WVCnoW-Alg}, then update the solution-storage accordingly.}
\end{reducebranchrule}

{\noindent By Lemma~\ref{lemma:wvc2}, we get a branching vector whose root is smaller than 1.3954.}

\bigskip

\myparagraph{Running Time} Let $R_{\Delta'}(\ell)$ denote the number of connected induced subgraphs of $G$ that contain $\ell$ vertices and have degree at most $\Delta'$. Robson \cite{memorization1} proved that $R_{\Delta'}(\ell) = \displaystyle{O^*((\frac{(\Delta'-1)^{\Delta'-1}}{(\Delta'-2)^{\Delta'-2}})^{\ell})}$. Thus, due to Rule \ref{rule:kernel}, we get that \alg{WVCnoW-Alg2} runs in time bounded by $O^*$ of

\[
\displaystyle{ \max_{1\leq k'\leq k} \!\left\{ 1.325^{k-k'}\!\!, \min\{1.3803^{k-k'}\!\!,R_4(4k'\!+\!1)\}, \min\{1.3954^{k-k'}\!\!,R_3(3k'\!+\!1)\} \!\right\} }
\]
\[= O^*(1.363^k)\]

\smallskip

\myparagraph{An $O^*(1.347^W)$ Time and Space Algorithm for {\sc WVC}} Let $(G=(V,E),w: V\rightarrow \mathbb{R}^{\geq 1},W)$ be the given instance of {\sc WVC}. Since a vertex cover of weight at most $W$ contains at most $\lfloor W\rfloor$ vertices, we can solve this instance by running \alg{WVCnoW-Alg2} on $(G=(V,E),w: V\rightarrow \mathbb{R}^{\geq 1},\lfloor W\rfloor)$, and returning its solution iff its weight is at most $W$. To obtain a running time bounded by $O^*(1.347^W)$, rather than $O^*(1.363^W)$, we replace Reduction Rule \ref{rule:kernel} by the following rule:

\setcounter{reducerule}{1}

\begin{reducerule}
{\normalfont [$|V'|> 2W'$]
Compute a subset $U\subseteq V$ of at most $2W^*$ vertices such that the weight of a minimum-weight vertex cover in both $G'$ and $G'[U]$ is $W^*$, as shown in the paper \cite{wvcker2008}. If $|U|>2W'$, return NIL; else, return \alg{WVCnoW-Alg2}$(G'[U],w',W')$.}
\end{reducerule}

{\noindent Then, we get that \alg{WVCnoW-Alg2} runs in time bounded by $O^*$ of}

\[
\displaystyle{ \max_{1\leq W'\leq \lfloor W\rfloor} \!\!\left\{\! 1.325^{\lfloor W\rfloor\!-\!W'}\!\!\!, \min\{1.3803^{\lfloor W\rfloor\!-\!W'}\!\!\!,R_4(2W')\}, \min\{1.3954^{\lfloor W\rfloor-W'}\!\!\!,R_3(2W')\} \!\right\} }
\]
\[= O^*(1.347^W)\]\qed
\end{proof}

\subsection{{\sc Restricted $k$-WVC} on Bipartite Graphs}\label{section:wvcbip}

Consider the following variant of {\sc WVC}, in which the parameter $k$ is used in the restricted manner described in Section \ref{section:technique}.

\myparagraph{Restricted $k$-WVC} Given an instance of {\sc WVC}, along with a parameter $k\in\mathbb{N}$, find a vertex cover of weight at most $W$ and
 size at most $k$. If such a vertex cover does not exist, return NIL.

\smallskip

The next result shows an advantage in solving at each iteration of an algorithm {\sc $k$-WVC} rather than {\sc Restricted $k$-WVC}. Clearly, {\sc $k$-WVC} is solvable in polynomial time on bipartite graphs, since {\sc WVC} is easy to solve on these graphs (see, e.g., \cite{Sch02}). For {\sc Restricted $k$-WVC}, however, this is not true:

\begin{theorem}\label{theorem:restrictedwvc}
{\sc Restricted $k$-WVC} on bipartite graphs is not in {\normalfont P} unless {\normalfont P=NP}.
\end{theorem}

\begin{proof}
We use a reduction from a variant of the following problem.

\myparagraph{Constrained Vertex Cover on Bipartite Graphs (Min-CVCB)} Given a bipartite graph $G=(L,R,E)$, and parameters $k_L\leq |L|$ and $k_R\leq |R|$, decide whether $G$ has a vertex cover consisting of at most $k_L$ vertices from $L$ and at most $k_R$ vertices from $R$.
\smallskip

Consider a variant of {\sc Min-CVCB}, that we call {\sc Min-CVCB$^*$}, where we need to decide whether $G$ has a vertex cover that consists of {\em exactly} $k_L$ vertices from $L$ and at most $k_R$ vertices from $R$. Since {\sc Min-CVCB} is is not in {\normalfont P} unless {\normalfont P=NP} \cite{constrainedvc2}, the same holds for {\sc Min-CVCB$^*$}.

Let $(G=(L,R,E), k_L, k_R)$ be an instance of {\sc Min-CVCB$^*$}. Let $n=|L\cup R|$. We define an instance $(G'=(L',R',E'),w':L'\cup R'\rightarrow \mathbb{R}^{\geq 1},W',k')$ of {\sc Restricted $k$-WVC} on bipartite graphs as follows.

\begin{itemize}
\item $L'=L$, and $\displaystyle{R'=R\cup (\bigcup_{v\in L}\{x^v_1,x^v_2,\ldots,x^v_{n^2}\})}$.

\item $\displaystyle{E'=E\cup (\bigcup_{v\in L}\{\{v,x^v_1\},\{v,x^v_2\},\ldots,\{v,x^v_{n^2}\}\})}$.

\item ($\forall v\in L': w'(v)=n^{10}$), and ($\forall v\in R': w'(v)=1$).

\smallskip
\item $W' = n^{10}k_L + k_R + n^2(|L|-k_L)$, and $k' = k_L + k_R + n^2(|L|-k_L)$.
\end{itemize}

Clearly, the above reduction is polynomial. Note that an illustrated example of the construction of $L',R'$ and $E'$ is given in Fig.~\ref{fig:reduction} (in this appendix). We next show that $(G,k_L,k_R)$ is a yes-instance iff $(G',w',W',k')$ is a yes-instance.

For one direction, suppose that $G$ has a vertex cover $S$ that consists of exactly $k_L$ vertices from $L$ and at most $k_R$ vertices from $R$. Let $S_L = S\cap L, S_R = S\cap R,$ and $\widetilde{S} = \{v\in R'\setminus R: N(v)\cap S_L=\emptyset\}$. Define $S'=S\cup\widetilde{S}$. Note that $S'$ is a vertex cover in $G'$, and $w'(S') = w'(S_L) + w'(S_R) + w'(\widetilde{S}) = n^{10}|S_L| + |S_R| +  n^2(|L|-|S_L|)\leq n^{10}k_L + k_R + n^2(|L|-k_L) = W'$. Moreover, $|S'| = |S_L| + |S_R| + |\widetilde{S}| \leq k_L + k_R + n^2(|L|-k_L) = k'$. Thus, $(G',w',W',k')$ is a yes-instance.

For the other direction, suppose that $G'$ has a vertex cover $S'$ such that $w'(S')\leq W'$ and $|S'|\leq k'$. Assume w.l.o.g that $S'$ is a minimal vertex cover. Let $S_L = S'\cap L, S_R = S'\cap R,$ and $\widetilde{S} = S'\setminus(L\cup R)$. Clearly, $|\widetilde{S}|=n^2(|L|-|S_L|)$. Since $n^{10}|S_L|\leq n^{10}|S_L| + |S_R| + n^2(|L|-|S_L|) = w'(S_L) + w'(S_R) + w'(\widetilde{S}) = w'(S') \leq W' = n^{10}k_L + k_R + n^2(|L|-k_L)\leq n^{10}k_L + n + n^3$, we have that $|S_L|\leq k_L$. Moreover, $|S_L| + |S_R| + n^2(|L|-|S_L|) = |S_L| + |S_R| + |\widetilde{S}| =
 |S'| \leq k' = k_L + k_R + n^2(|L|-k_L)$. Thus, $|S_R| + (n^2-1)k_L \leq k_R + (n^2-1)|S_L|$. This implies that $k_L\leq |S_L|$. Thus, we conclude that $|S_L|=k_L$, which, by the above argument, further implies that $|S_R|\leq k_R$. Define $S=S_L\cup S_R$. Since $S$ is a vertex cover in $G$, we get that $(G,k_L,k_R)$ is a yes-instance.\qed
\end{proof}

\begin{figure}[!ht]\centering
\frame{\includegraphics[scale=0.8]{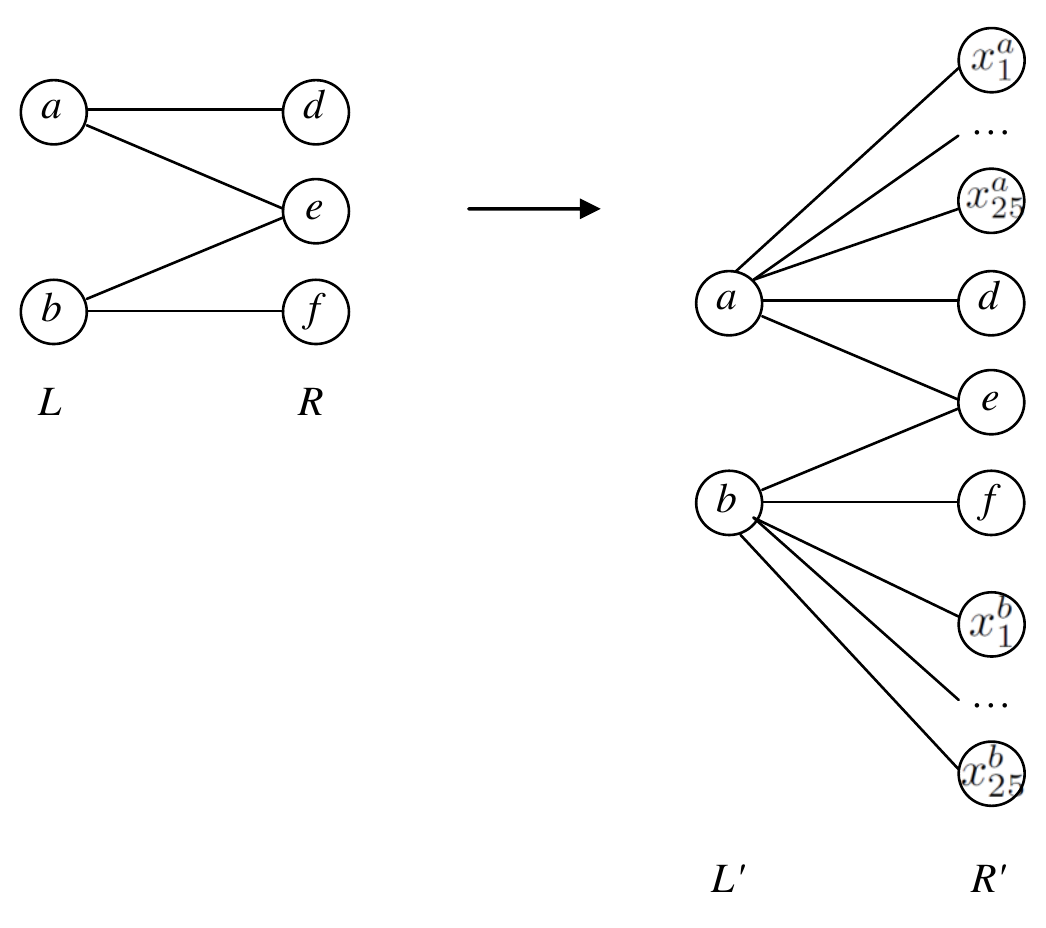}}
\caption{An example of the construction in the proof of Theorem \ref{theorem:restrictedwvc}.}\label{fig:reduction}
\end{figure}
\subsection{{\sc WVC} Parameterized by the Size of a Minimum VC}\label{section:wvc3}

In this appendix, we observe that the algorithm of  \cite{vc2007} for {\sc VC} can be modified to solve {\sc WVC} in $O^*(1.443^t)$ time and polynomial space. This algorithm is based on the measure and conquer technique \cite{newdfsbook}. In this variant of the bounded search tree technique, one uses a non-standard measure to analyze the branching vectors associated with the branching rules of the algorithm.\footnote{The algorithm of \cite{vc2007} is also based on the iterative compression technique (see \cite{newdfsbook}), which we replace by a simple call to the currently best algorithm for {\sc VC}, given~in~\cite{vc2010}.} Apart from this observation, our contribution also lies in introducing a preprocessing phase and new rules. Combined with a corresponding refined analysis of the algorithm, this results in a faster running time of $O^*(1.415^t)$ on graphs of bounded degree 3. Interestingly, given a vertex cover $U$ of the input graph $G=(V,E)$, the rules of the algorithm in \cite{vc2007} rely only on the structure of $G[U]$, while our improvement relies on the {\em relation} between $G[U]$ and $G[V\setminus U]$.

We develop below an algorithm, \alg{WVC*-Alg}, that solves the following variant of {\sc WVC}.

\myparagraph{WVC*} Given an instance of {\sc WVC}, along with a minimum(-size) vertex cover $U$, return a vertex cover of weight at most $W$ (if one exists).

\smallskip

For this algorithm, we prove the following.

\begin{theorem}
\alg{WVC*-Alg} solves {\sc WVC*} in $O^*(1.443^t)$ time and polynomial space. On graphs of bounded degree 3, \alg{WVC*-Alg} solves {\sc WVC*} in $O^*(1.415^t)$ time and polynomial space.
\end{theorem}

Given an instance $(G,w,W)$ of {\sc WVC}, we can find, in $O^*(1.274^t)$ time and polynomial space, a minimum vertex cover \cite{vc2010}. Thus, we obtain the following.

\begin{cor}
{\sc WVC} can be solved in $O^*(1.443^t)$ time and polynomial space. On graphs of bounded degree 3, {\sc WVC} can be solved in $O^*(1.415^t)$ time and polynomial space.
\end{cor}

Let $C(G,U)$ be the set of connected components in $G[U]$. Also, let $C_3(G,U)$, $P_2(G,U)$ and $S(G,U)$ be the set of cycles on (exactly) 3 vertices (i.e., triangles), paths on (exactly) 2 vertices, and single vertices in $C(G,U)$, respectively. Let $C_3^*(G,U)$ contain each triangle $c_3\in C_3(G,U)$ such that the vertices in $V(c_3)$ do not have a common neighbor in $V\setminus U$ (i.e., $\bigcap_{v\in V(c_3)}N(v)=(\bigcap_{v\in V(c_3)}N(v))\setminus U=\emptyset$). We say that a minimum vertex cover $U$ is {\em good} if there exists a function $f: C_3^*(G,U)\rightarrow P_2(G,U)$, such that for each $c_3\in C_3^*(G,U)$, there is a vertex in $V\setminus U$ that is a neighbor of a vertex in $c_3$ and both vertices in $f(c_3)$. Note that an example of a good minimum vertex cover is illustrated in Fig.~\ref{fig:f}. In this example, $C(G,U)=\{A,B,C,D,E\}, S(G,U)=\emptyset, C_3(G,U)=\{A,C,E\}, P_2(G,U)=\{B,D\}$ and $C_3^*(G,U)=\{C,E\}$. A function $f$, which shows that $U$ is good, can assign $f(C)=B$ or $f(C)=D$, and $f(E)=D$.

\begin{figure}[!ht]\centering
\frame{\includegraphics[scale=0.8]{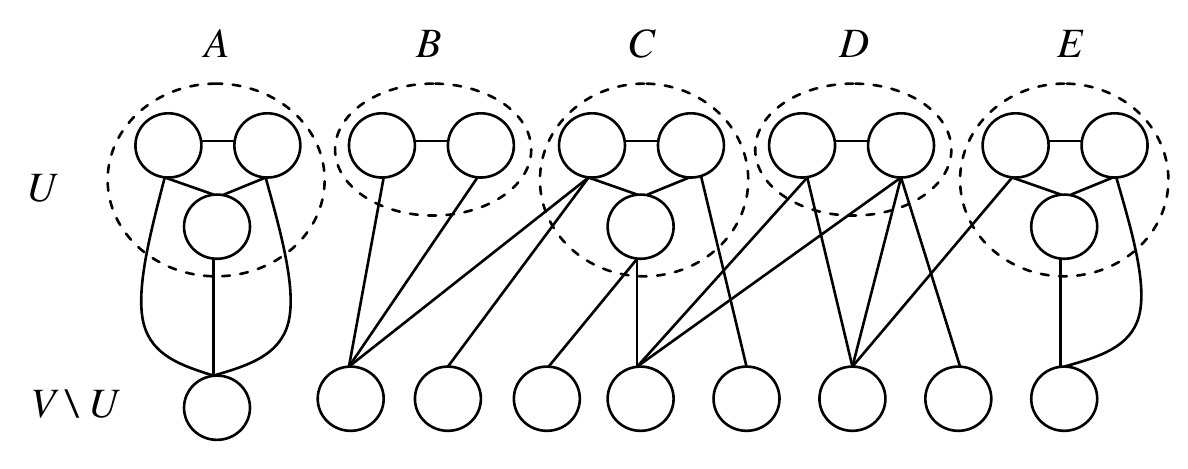}}
\caption{An example of a good minimum vertex cover $U$.}\label{fig:f}
\end{figure}

If $G$ is a graph of bounded degree 3, \alg{WVC*-Alg} first executes a preprocessing phase where it replaces $U$ by a good vertex cover (of the same size) and obtains a corresponding function $f$, using the following result.

\begin{lem}\label{lemma:vcpreprocessing}
Given an instance of {\sc WVC*}, where $G$ is a graph of bounded degree 3, a good minimum vertex cover, along with a corresponding function $f$, can be computed in polynomial time.
\end{lem}

\begin{proof}
Let $c_3$ be a triangle in $C_3^*(G,U)$, and let $v$ be a vertex in $V(c_3)$. Note that, since $G$ has maximum degree 3, $|N(v)\setminus V(c_3)|\leq 1$. Therefore, $N(v)\setminus V(c_3)$ contains exactly one vertex (which belongs to $V\setminus U$), that we denote by $u$, since otherwise we can remove $v$ from $U$ and obtain a vertex cover smaller than $U$ (although $U$ is a {\em minimum} vertex cover). Suppose that $N(u)\setminus \{v\}$ does not contain both vertices of a path in $P_2(G,U)$. Then, we replace $v$ by $u$ (i.e., remove $v$ from $U$, and insert $u$ to $U$), and obtain a minimum vertex cover $U'$. Now, $c_3\notin C_3^*(G,U')$. Note that, since $c_3\in C_3^*(G,U)$, we have that $N(u)\neq V(c_3)$. For this reason, and since $N(u)\setminus \{v\}$ does not contain both vertices of a path in $P_2(G,U)$, we get that $C_3^*(G,U')\subset C_3^*(G,U)$.

Repeating the above argument a polynomial number of times, we obtain a minimum vertex cover $U^*$, such that for any $c_3\in C_3^*(G,U^*)$ and $v\in V(c_3)$, letting $u$ be the vertex in $N(v)\setminus V(c_3)$, there is a path $p_2\in P_2(G,U^*)$ such that $N(u)\setminus\{v\} = V(p_2)$. Thus, we obtain the required function $f: C_3^*(G,U^*)\rightarrow P_2(G,U^*)$.\qed
\end{proof}

A call to \alg{WVC*-Alg} is of the form \alg{WVC*-Alg}$(G,w,W,U,f)$, where, if $G$ is a graph of degree larger than 3, $f=NIL$. In this algorithm, we analyze branching vectors with respect to the measure $m(G,U)=|U|-|S(G,U)|$. Initially, $m(G,U)\leq |U|=t$. Since the worst root we get is bounded by 1.443, where for graphs of bounded degree 3, it is further bounded by 1.415, and since the branching stops if $m(G,U)\leq 0$, we obtain the desired running time. Note that, if $G$ is a graph of bounded degree 3, removing vertices from $U$ cannot add triangles to $C_3(G,U)$, since this implies that the original input vertex cover did not have minimum size. Also, if $G$ is a graph of bounded degree 3, when we remove a path $p_2$ from $P_2(G,U)$, we ensure that we also remove the triangles that $f$ maps to it (see Rules 4 and 5). Thus, throughout the execution, we can use the function $f$ that was computed in the preprocessing phase.\footnote{We note that, during the execution, $U$ remains a vertex cover, though it may not remain a minimum vertex cover.}

We now present each rule related to a call \alg{WVC*-Alg}$(G,w,W,U,f)$.

\setcounter{reducerule}{0}

\begin{reducerule}\label{rule:bipartite}
{\normalfont [$G$ is a bipartite graph]
Compute a minimum-weight vertex cover $A$ of $G$ (see, e.g., \cite{Sch02}). Return $A$ iff $w(A)\leq W$.}
\end{reducerule}

{\noindent If $|U|=S(G,U)$, $G$ is a bipartite graph. Thus, if $m(G,U)$ is decreased to 0 (it cannot be negative), the algorithm stops branching, since the condition in this rule is true.\footnote{Recall that such an observation is required in the analysis of the running time of a bounded search tree-based algorithm in the manner described in Section \ref{section:preliminaries}.}}

\begin{reducerule}\label{rule:remove}
{\normalfont [There is $v\!\in\! U$ such that $N(v)\setminus U\!=\!\emptyset$]
Return \alg{WVC*-Alg}($G,$ $w,W,U\setminus \{v\},f$).}
\end{reducerule}

{\noindent In this rule, we can simply remove $v$ from $U$. Indeed, since $N(v)\setminus U\!=\!\emptyset$, $U\setminus\{v\}$ is a vertex cover. If $v\notin S(G,U)$, $m(G,U)$ is decreased by 1.}

\begin{reducerule}
{\normalfont [There is a connected component $H$ such that $|V(H)|\!\leq\! 10$]
Use brute-force to compute a minimum-weight vertex cover $A$ of $H$. Return \alg{WVC*-Alg}$(G[V\setminus V(H)], w, W\!-\!w(A), U\!\setminus\! V(H),f)\cup A$.}
\end{reducerule}

{\noindent Clearly, this rule is correct. We note that, in particular, if $G$ is a graph of bounded degree 3, this rule eliminates triangles in $C_3(G,U)\setminus C_3^*(G,U)$. Indeed, in this case, a triangle $c_3\in C_3^*(G,U)$, along with the common neighbor of its vertices, form a connected component on 4 vertices; therefore, this rule removes it from $G$. Thus, in the following rules, if $G$ is a graph of bounded degree 3, $C_3^*(G,U)=C_3(G,U)$.}

\bigskip

{\noindent\bf Apply the following two rules only if the maximum degree of the original input graph is bounded by 3.}

\begin{branchrule}
{\normalfont [There is $c_3\in C_3^*(G,U)$ such that $(\forall c_3'\in C_3^*(G,U)\setminus\{c_3\}: f(c_3)\neq f(c_3'))$]
Let $v$ be a vertex in $V(f(c_3))$.
\begin{enumerate}
\item For all $A\subseteq V(c_3)$ such that $|A|=2$:
\smallskip
	\begin{itemize}
 \item If the result of \alg{WVC*-Alg}($G[V\setminus X],w,W\!-w(X),U\setminus X,f$), where $X=\{v\}\cup A$, is not NIL: Return it along with $\{v\}\cup A$.
	\end{itemize}
\smallskip
\item Else: Return \alg{WVC*-Alg}($G[V\setminus N(v)],w,W\!-w(N(v)),U\setminus N(v),f)\cup N(v)$.
\end{enumerate}}
\end{branchrule}

{\noindent In this rule, we choose either $v$ (in the branches in the first item) or $N(v)$, where if we choose $v$, we try every option of choosing two vertices from $c_3$. Therefore, since we must choose at least two vertices of a triangle to a vertex cover, the branching is exhaustive. Choosing $v$ and a set $A\subseteq V(c_3)$ such that $|A|=2$, we decrease $m(G,U)$ by 5 (since $U$ decreases by $|X|=3$, and $S(G,U)$ increases by $|(V(f(c_3))\cup V(c_3))\setminus X|=2$). Choosing $N(v)$, we next apply Rule \ref{rule:remove} on a vertex of $c_3$, and thus overall decrease $m(G,U)$ by $|V(f(c_3))|+1=3$. We get the branching vector $(5,5,5,3)$, whose root is smaller than 1.415.}

\begin{branchrule}
{\normalfont [There are different $c_3,c_3'\in C_3^*(G,U)$ such that $f(c_3)=f(c_3')$]
Let $v$ be a vertex in $V(f(c_3))$.
\begin{enumerate}
\item For all $A\subseteq V(c_3),B\subseteq V(c_3')$ such that $|A|=|B|=2$:
\smallskip
	\begin{itemize}
 \item If the result of \alg{WVC*-Alg}($G[V\setminus X],w,W\!-w(X),U\setminus X,f$), where $X=\{v\}\cup A\cup B$, is not NIL: Return it along with $\{v\}\cup A\cup B$.
 \end{itemize}
\smallskip
\item Else: Return \alg{WVC*-Alg}($G[V\setminus N(v)],w,W\!-w(N(v)),U\setminus N(v),f)\cup N(v)$.
\end{enumerate}}
\end{branchrule}

{\noindent As in the previous rule, since we must choose at least two vertices of a triangle to a vertex cover (here we consider the triangles $c_3$ and $c_3'$), this branching is exhaustive. Choosing $v$ and sets $A$ and $B$ as specified in the rule, we decrease $m(G,U)$ by $|V(f(c_3))\cup V(c_3)\cup V(c_3')|=8$ (since $U$ decreases by $|X|=5$, and $S(G,U)$ increases by $|(V(f(c_3))\cup V(c_3)\cup V(c_3'))\setminus X|=3$). Choosing $N(v)$, we next apply Rule \ref{rule:remove} on a vertex of $c_3$ and a vertex of $c_3'$, and thus decrease $m(G,U)$ by $|V(f(c_3))|+2=4$. We get the branching vector $(8,8,8,8,8,8,8,8,8,4)$, whose root is smaller than 1.415.}

\bigskip

{\noindent Note that, since the previous three rules did not apply, we next assume that if the input graph has maximum degree 3, $C_3(G,U)=\emptyset$. In the remaining (branching) rules, we first branch on neighbors of leaves in $G[U]$ whose degree in $G[U]$ is at least two, then on leaves in $G[U]$, and finally on the remaining vertices in $U\setminus V(S(G,U))$. All of these rules are exhaustive. Although we can merge some of them, we present them separately for the sake of clarity.}

\begin{branchrule}\label{rule:wvcneileaf}
{\normalfont [There are $v,u\in U$ s.t.~$N(u)\cap U = \{v\}$ and $|N(v)\cap U|\geq 2$]
\begin{enumerate}
\item If the result of \alg{WVC*-Alg}($G[V\setminus \{v\}],w,W\!-w(v),U\setminus \{v\},f$) is not NIL: Return it along with $v$.
\item Else: Return \alg{WVC*-Alg}($G[V\setminus N(v)],w,W\!-w(N(v)),U\setminus N(v),f)\cup N(v)$.
\end{enumerate}}
\end{branchrule}

{\noindent It is easy to see that we get a branching vector that is at least as good as $(|\{v,u\}|,|(N(v)\cap U)\cup\{v\}|)$. Indeed, in the first branch, $v$ is removed from $U$ and $u$ is inserted to $S(G,U)$, and in the second branch, $N(v)\cap U$ is removed from $U$ and $v$ is inserted to $S(G,U)$. Since this branching vector is at least as good as $(2,3)$ (since $|N(v)\cap U|\geq 2$), we get a root that is smaller than 1.415.}

\begin{branchrule}\label{rule:wvcneileaf2}
{\normalfont [There are $v,u\in U$ such that $N(v)\cap U = \{u\}$]
\begin{enumerate}
\item If the result of \alg{WVC*-Alg}($G[V\setminus \{v\}],w,W\!-w(v),U\setminus \{v\},f$) is not NIL: Return it along with $v$.
\item Else: Return \alg{WVC*-Alg}($G[V\setminus N(v)],w,W\!-w(N(v)),U\setminus N(v),f)\cup N(v)$.
\end{enumerate}}
\end{branchrule}

{\noindent Since the previous rule did not apply, $N(u)\cap U=\{v\}$. Thus, at each branch, one vertex in $\{v,u\}$ is removed from $U$, and the other is inserted to $S(G,U)$. We get the branching vector $(|\{v,u\}|,|\{v,u\}|)$, whose root is smaller than 1.415.}

\begin{branchrule}\label{rule:tvc3}
{\normalfont [There is $v\in U$ such that $|N(v)\cap U| \geq 3$]
\begin{enumerate}
\item If the result of \alg{WVC*-Alg}($G[V\setminus \{v\}],w,W\!-w(v),U\setminus \{v\},f$) is not NIL: Return it along with $v$.
\item Else: Return \alg{WVC*-Alg}($G[V\setminus N(v)],w,W\!-w(N(v)),U\setminus N(v),f)\cup N(v)$.
\end{enumerate}}
\end{branchrule}

{\noindent Clearly, we get the branching vector $(|\{v\}|,|(N(v)\cap U)\cup\{v\}|)$. This branching vector is at least as good as $(1,4)$ (since $|N(v)\cap U| \geq 3$), and thus we get a root that is smaller than 1.415.}

\begin{branchrule}
{\normalfont [There is $c_3\in C_3(G,U)$] For all $A\subseteq V(c_3)$ s.t.~$|A|=2$:
\begin{itemize}
\item If the result of \alg{WVC*-Alg}($G[V\setminus A],w,W\!-w(A),U\setminus A,f$) is not NIL: Return it along with $A$.
\end{itemize}
\noindent If none of the branches returned a result that is not NIL: Return NIL.}
\end{branchrule}

{\noindent Recall that, at this point, if the original input graph has maximum degree 3, $C_3(G,U)=\emptyset$. Thus, this rule is applied only if the original input graph has degree greater than 3. At each branch, two vertices of $c_3$ are removed from $U$, and the other one is inserted to $S(G,U)$. Thus, we get the branching vector $(3,3,3)$, whose root is smaller than 1.443. Clearly, after this rule, $C_3(G,U)=\emptyset$.}

\begin{branchrule}
{\normalfont [Remaining case]
Let $v$ be a vertex in $U$ s.t.~$|N(v)\cap U|\!=\!2$.
\begin{enumerate}
\item If the result of \alg{WVC*-Alg}($G[V\setminus \{v\}],w,W\!-w(v),U\setminus \{v\},f$) is not NIL: Return it along with $v$.
\item Else: Return \alg{WVC*-Alg}($G[V\setminus N(v)],w,W\!-w(N(v)),U\setminus N(v),f)\cup N(v)$.
\end{enumerate}}
\end{branchrule}

{\noindent In this rule, $G[U]$ does not include vertices of degree at least 3 (due to Rule \ref{rule:tvc3}), triangles (since, as we concluded in the previous rule, $C_3(G,U)=\emptyset$) or leaves (due to Rules \ref{rule:wvcneileaf} and \ref{rule:wvcneileaf2}). Thus, $C(G,U)\setminus S(G,U)$ is a set of cycles, each on at least 4 vertices. Thus, after choosing $v$, we can apply Rule \ref{rule:wvcneileaf}. This results in a branching vector that is at least as good as $(1+(2,3),3)=(3,4,3)$, whose root is smaller than 1.415.}

\section{Weighted 3-Hitting Set}\label{section:w3hs}

In this appendix, we develop an algorithm for {\sc W3HS} which uses $O^*(2.168^s)$ time and polynomial space (see Appendix \ref{section:whs1}). This algorithm is complemented (in Appendix \ref{section:whs2}) by an algorithm for {\sc W3HS} which uses $O^*(1.381^{s-t}2.381^t)$ time and polynomial space, or $O^*(1.363^{s-t}2.363^t)$ time and $O^*(1.363^s)$~space.

\subsection{An $O^*(2.168^s)$-Time Algorithm for {\sc $k$-W3HS}}\label{section:whs1}

This appendix presents \alg{W3HS-Alg}, an algorithm for {\sc $k$-W3HS}, based on measure and conquer (see Appendix \ref{section:wvc3}). Some of our rules build upon the rules given in \cite{whs2010}, applied in a more refined manner, and an adaptation of the measure used in \cite{hs2007}. We obtain the~following~result.

\begin{theorem}
\alg{WHS-Alg} solves {\sc $k$-W3HS} in $O^*\!(2.168^k\!)$ time and polynomial space.
\end{theorem}

By the discussion in Section \ref{section:technique}, this implies the desired result:

\begin{cor}
{\sc W3HS} can be solved in $O^*(2.168^s)$ time and polynomial space.
\end{cor}

Next, we present each rule related to a call \alg{WHS-Alg}$(G,w,W,k)$. When presenting a branching rule, we analyze its branching vector with respect to $m(G,k)=k-\alpha(G)+1$, where
\begin{enumerate}
\item If $G$ contains at least four 2-edges: $\alpha(G)=\alpha_4=0.87$.
\item Else if $G$ contains exactly three 2-edges, and no vertex is contained in all of them: $\alpha(G)\!=\!\alpha_3\!=\!0.8$.
\item Else if $G$ contains at least two 2-edges: $\alpha(G)=\alpha_2=0.55$.
\item Else if $G$ contains exactly one 2-edge: $\alpha(G)=\alpha_1=0.35$.
\item Else ($G$ does not contain a 2-edge): $\alpha(G)=0$.
\end{enumerate}

We note that the constants above were chosen to allow us to obtain branching vectors whose roots are bounded by 2.168,\footnote{That is, we examined the transitions between the cases associated with the constants (performed when executing our branching rules), and optimized the constants accordingly.} yielding the desired running time.

\setcounter{reducerule}{0}

\begin{reducerule}
{\normalfont [$\min\{W,k\}<0$]
Return NIL.}
\end{reducerule}

{\noindent Thus, the branching stops when $m(G,k)\leq 0$ (it actually stops earlier, which only improves the running time of the algorithm). To see this, note that if $m(G,k)\leq 0$, then $k< k-\alpha(G)+1 = m(G,k)\leq 0$.}

\begin{reducerule}\label{rule2:hsdeg2}
{\normalfont [$G$ is a hypergraph of bounded degree 2]\footnote{Recall that we defined the degree of a vertex in a hypergraph in Section \ref{section:preliminaries}.}
Compute a minimum-weight hitting set $U$ of $G$ in polynomial time (see \cite{whs2010}). Return $U$ iff $w(U)\leq W$.}
\end{reducerule}

{\noindent Clearly, if there is a hitting set of weight at most $W$, it is returned by this rule. Note that the correctness of this rule relies on our flexible use of the parameter $k$, since we may return a hitting set that contains more than $k$ vertices.}

\begin{reducerule}\label{rule:whsedgedominate}
{\normalfont [There are $e,e'\in E$ s.t.~$e\subset e'$]
Return \alg{WHS-Alg}($G'=(V, E\setminus \{e'\}),w,W,k$).}
\end{reducerule}

{\noindent This is a standard edge-domination rule: If $e$ is covered, so is $e'$ (because $e\subset e'$), and since a hitting set of a graph covers all of its edges, we can simply omit $e'$. The application of this rule does not change $m(G,k)$.}

\begin{reducerule}\label{rule:whsdominate}
{\normalfont [There are $v,u\in V$ such that $w(v)\leq w(u)$ and $E(u)\subseteq E(v)$] Consider the following cases.
\begin{enumerate}
\item If $\{v,u\}\in E$: Return \alg{WHS-Alg}($G'=(V\setminus \{v,u\}, E\setminus E(v)),w,W\!-\!w(v),k\!-\!1$), along~with~$v$.
\item Else: Return \alg{WHS-Alg}($G[V\setminus \{u\}],w,W,k$).
\end{enumerate}}
\end{reducerule}

{\noindent This is a standard vertex-domination rule: Since $w(v)\leq w(u)$, and $v$ covers all the edges that $u$ covers, it is always better to choose $v$ rather than $u$ (in the first case, we choose $v$, and not only omit $u$, to cover the edge $\{v,u\}$). The application of this rule does not increase $m(G,k)$, since if $\alpha(G)$ increases, the increase is smaller than 1, while $k$ decreases by 1.}

\begin{branchrule}\label{rule:02edge}
{\normalfont [$G$ does not contain a 2-edge] Let $v$ be a vertex of maximum degree in $G$. 
\begin{enumerate}
\item If the result of \alg{WHS-Alg}($G'=(V\!\setminus\! \{v\}, E\!\setminus\! E(v)),w,W\!\!-\!w(v),k\!-\!1$) is not NIL: Return it along~with~$v$.
\item Else: Return \alg{WHS-Alg}($G[V\setminus \{v\}],w,W,k$).
\end{enumerate}}
\end{branchrule}

{\noindent This branching is exhaustive ($v$ is either chosen or omitted), and therefore correct. Clearly, in the first branch we choose $v$, and thus $m(G,k)$ is decreased by 1. It is easy to see that $G[V\setminus\{v\}]$ contains (at least) three 2-edges such that no vertex is contained in all of them. Indeed, suppose by way of contradiction that this is not the case. Then, there is a vertex $u\in V$ such that every edge $e\in E$ satisfying $v\in e$, also satisfies $u\in e$. Since $v$ is a vertex of maximum degree in $G$, every $e\in E$ satisfying $u\in e$, also satisfies $v\in e$. This is a contradiction, since Rule \ref{rule:whsdominate} precedes the current rule (simply let $v$, defined in Rule \ref{rule:whsdominate}, be the vertex of minimum weight among $v$ and $u$ to which we refer in this explanation). Therefore, in the second branch, $\alpha(G)$ increases from 0 to (at least) $\alpha_3$. Thus, we get a branching vector that is at least as good as $(1,\alpha_3)=(1,0.8)$, whose root is smaller than 2.168.}

\bigskip

{\noindent From now on, since Rule \ref{rule:02edge} did not apply, $G$ contains at least one~2-edge.}

\begin{branchrule}\label{rule:12edgea}
{\normalfont [$G$ contains exactly one 2-edge $\{v,u\}$, and $|E(v)|\geq 3$]
\begin{enumerate}
\item If the result of \alg{WHS-Alg}($G'\!=\!(V\!\setminus\! \{v\}, E\!\setminus\! E(v)),w,W\!\!\!-\!w(v),k\!-\!1$) is not NIL: Return it along~with~$v$.
\item Else: Return \alg{WHS-Alg}($G'[V\setminus\{v,u\}],w,W\!\!-\!w(u),k-1)\cup\{u\}$, where $G' = G(V\!\setminus\! \{u\}, E\!\setminus\! E(u))$.
\end{enumerate}}
\end{branchrule}

{\noindent This branching is exhaustive. Choosing $v$, we delete the 2-edge $\{v,u\}$. Choosing $u$, we also delete this edge, but introduce (at least) two new 2-edges: There are at least two 2-edges in $E(v)$ (since $|E(v)\setminus\{\{v,u\}\}|\geq 2$), none of them contains $u$ (by Rule \ref{rule:whsedgedominate}) or a 2-edge (by Rule \ref{rule:whsedgedominate}, but also, since in this case, the only 2-edge is $\{u,v\}$); thus, once we omit $v$ in the second branch (where we choose $u$), the 3-edges in $E(v)$ turn into new 2-edges. Therefore, the branching vector is at least as good as $(1\!-\!\alpha_1,1\!-\!\alpha_1\!+\!\alpha_2)=(0.65,1.2)$, whose root is smaller than 2.168.}

\begin{branchrule}\label{rule:12edgeb}
{\normalfont [$G$ contains exactly one 2-edge $\{v,u\}$, where $w(v)\geq w(u)$] Let $e$ be the 3-edge adjacent to $v$. Let $S$ is the set of vertices of 1-edges in $\widetilde{G}[V\setminus e]$, where $\widetilde{G}=(V\setminus \{v\}, E\setminus E(v))$.
\begin{enumerate}
\item If the result of \alg{WHS-Alg}($G'[V\setminus (e\cup S)],w,W-w(\{v\}\cup S),k-1-|S|$) is not NIL, where $G'=(V\setminus (\{v\}\cup S), E\setminus E(\{v\}\cup S))$: Return it along with $\{v\}\cup S$.
\item Else: Return \alg{WHS-Alg}($G'[V\setminus \{v,u\}],w,W\!-w(u),k-1)\cup\{u\}$, where $G'=(V\!\setminus\! \{u\}, E\!\setminus\! E(u))$.
\end{enumerate}}
\end{branchrule}

{\noindent Since this rule is slightly technical, it is illustrated in Fig.~\ref{fig:12edgeb}. First, note that there exists a 3-edge adjacent to $v$ (by Rule \ref{rule:whsdominate}), and there is no more than one such 3-edge (by Rule \ref{rule:12edgea}), and therefore the choice of $e$ is well-defined. We next consider the correctness and branching vector of this rule. Denote by $r$ a vertex in $e\setminus \{v\}$ of degree at least 2. Note that this choice is possible, since if both vertices in $e\setminus \{v\}$ were of degree 1, Rule \ref{rule:whsdominate} would have been applied. Choosing $v$ (in the first branch), we need not choose another vertex from $e$, since then we can replace $v$ by $u$ and get a hitting set that is not heavier than the one with $v$ (note that the choice of $u$ is examined in the second branch). Thus, since we do not choose another vertex from $e$, we must choose all the vertices in $S$ (to get a hitting set). In this case, we delete $\{v,u\}$, but either $|S|\geq 1$ or we introduce a new 2-edge, since there is a 3-edge adjacent to $r$ (and it turns into a 1-edge or a 2-edge). Omitting $v$ (and choosing $u$), we delete $\{v,u\}$, but introduce a new 2-edge (from $e$). Therefore, at worst, we get the branching vector $(1,1)$, whose root~is~2.}

\begin{figure}[!h]\centering
\frame{\includegraphics[scale=0.8]{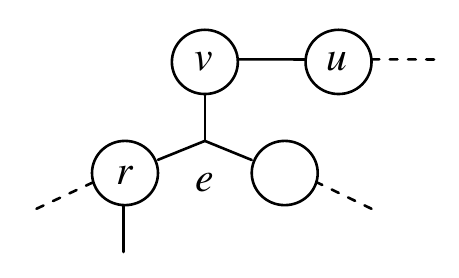}}
\caption{Rule \ref{rule:12edgeb} of \alg{WHS-Alg}.}
\label{fig:12edgeb}
\end{figure}

\bigskip

{\noindent From now on, since Rules \ref{rule:12edgea} and \ref{rule:12edgeb} did not apply, $G$ contains at least two 2-edges.}

\begin{branchrule}\label{rule:star}
{\normalfont [$G$ contains exactly three 2-edges, and they have a common vertex $v$] Let $v_1,v_2$ and $v_3$ be the other vertices in the 2-edges.
\begin{enumerate}
\item If the result of \alg{WHS-Alg}($G'\!=\!(V\!\setminus\! \{v\}, E\!\setminus\! E(v)),w,W\!\!-\!w(v),k\!-\!1$) is not NIL: Return it along~with~$v$.
\item Else: Return \alg{WHS-Alg}($G'[V\setminus \{v,v_1,v_2,v_3\}],w,W-w(\{v_1,v_2,v_3\}),k-3)\cup\{v_1,v_2,v_3\}$, where $G'=(V\setminus \{v_1,v_2,v_3\}, E\setminus (E(v_1)\cup E(v_2)\cup E(v_3))$.
\end{enumerate}}
\end{branchrule}

{\noindent This branching is exhaustive: we either choose $v$, or delete $v$ and therefore must choose $v_1,v_2$ and $v_3$ (to get a hitting set). At worst, no 2-edge is introduced (in both branches), and thus we get a branching vector that is at least as good as $(1-\alpha_2,3-\alpha_2)=(0.45,2.45)$, whose root is smaller than 2.168.}

\bigskip

{\noindent From now on, since Rule \ref{rule:star} did not apply, $G$ contains exactly two 2-edges, or exactly three 2-edges where there is no vertex that is contained in all of them, or at least four 2-edges.}

\begin{branchrule}\label{rule:hs9}
{\normalfont [$G$ contains exactly two 2-edges, and there exist different $v_1,v,v_2\in V$ such that $\{v_1,v\},\{v,v_2\}\in E$ and $|E(v)|\geq 3$]
\begin{enumerate}
\item If the result of \alg{WHS-Alg}($G'\!=\!(V\!\setminus\! \{v\}, E\!\setminus\! E(v)),w,W\!\!-\!w(v),k\!-\!1$) is not NIL: Return it along~with~$v$.
\item Else: Return \alg{WHS-Alg}($G'[V\setminus \{v_1,v,v_2\}],w,W\!-w(\{v_1,v_2\}),k-2)\cup \{v_1,v_2\}$, where $G'=(V\setminus \{v_1,v_2\}, E\setminus (E(v_1)\cup E(v_2)))$.
\end{enumerate}}
\end{branchrule}

{\noindent Note that this rule is illustrated in Fig.~\ref{fig:hs9}. The branching is exhaustive (we either choose $v$, or delete $v$ and thus choose $v_1$ and $v_2$). Choosing $v$, we delete both existing 2-edges. Deleting $v$, we also delete both existing 2-edges, but introduce a new 2-edge (from a 3-edge previously adjacent to $v$). We get a branching vector at least as good as $(1-\alpha_2,2-\alpha_2+\alpha_1)=(0.45,1.8)$, whose root is smaller than~2.168.}

\begin{figure}[!h]\centering
\frame{\includegraphics[scale=0.8]{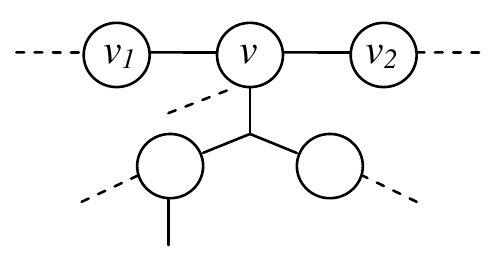}}
\caption{Rule \ref{rule:hs9} of \alg{WHS-Alg}.}
\label{fig:hs9}
\end{figure}

\begin{branchrule}\label{rule:hs10}
{\normalfont [$G$ contains exactly two 2-edges, and there exist different $v_1,v,v_2\in V$ such that $\{v_1,v\},\{v,v_2\}\in E$, and $|E(v_1)|\geq 3$]
\begin{enumerate}
\item If the result of \alg{WHS-Alg}($G'\!=\!(V\!\setminus\! \{v_1\}, E\!\setminus\! E(v_1)),w,W\!\!-\!w(v_1),k\!-\!1$) is not NIL: Return it along~with~$v_1$.
\item Else: Return \alg{WHS-Alg}($G'[V\setminus \{v_1,v\}],w,W\!-w(v),k-1)\cup\{v\}$, where $G'=(V\setminus \{v\}, E\setminus E(v))$.
\end{enumerate}}
\end{branchrule}

{\noindent This branching is exhaustive (we either choose $v_1$, or delete $v_1$ and thus choose $v$). Choosing $v_1$, we delete one 2-edge. Choosing $v$, we delete two 2-edges, but also introduce two new 2-edges (from 3-edges previously adjacent to $v_1$). We get a branching vector at least as good as $(1-\alpha_2+\alpha_1,1)=(0.8,1)$, whose root is smaller than 2.168.}

\begin{figure}[!h]\centering
\frame{\includegraphics[scale=0.8]{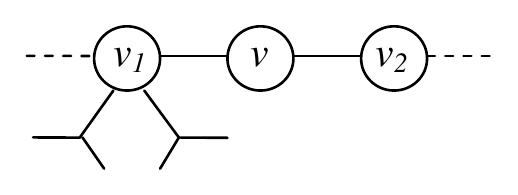}}
\caption{Rule \ref{rule:hs10} of \alg{WHS-Alg}.}
\end{figure}

\begin{branchrule}\label{rule:hs11}
{\normalfont [$G$ contains exactly two 2-edges, and there exist different $v_1,v,v_2\in V$ such that $\{v_1,v\},\{v,v_2\}\in E$, $|E(v_1)|=2$ and $w(v_1)\geq w(v)$] Let $e$ be the 3-edge adjacent to $v_1$. Let $S$ is the set of vertices of 1-edges in $\widetilde{G}[V\setminus e]$, where $\widetilde{G}=(V\setminus \{v_1\}, E\setminus E(v_1))$.
\begin{enumerate}
\item If the result of \alg{WHS-Alg}($G'[V\setminus (e\cup S)],w,W\!-w(\{v_1\}\cup S),k-1-|S|$) is not NIL, where $G'=(V\setminus (\{v_1\}\cup S), E\setminus E(\{v_1\}\cup S))$: Return it along with $\{v_1\}\cup S$.
\item Else: Return \alg{WHS-Alg}($G'[V\setminus \{v_1,v\}],w,W\!-w(v),k-1)\cup\{v\}$, where $G'=(V\setminus \{v\}, E\setminus E(v))$.
\end{enumerate}}
\end{branchrule}

{\noindent As in the proof of Rule \ref{rule:12edgeb}, there exists $r\in (e\setminus \{v_1\})$ of degree at least 2. Choosing $v_1$, we need not choose another vertex from $e$, since then we can replace $v_1$ by $v$ and get a hitting set that is not heavier than the one with $v_1$ (the choice of $v$ is examined in the second branch). At worst, in the first branch, $S=\emptyset$; then, we delete $\{v_1,v\}$ (but not $\{v,v_2\}$), and introduce a new 2-edge (from a 3-edge previously adjacent to $r$). Deleting $v_1$ (and choosing $v$), we delete $\{v_1,v\}$ and $\{v,v_2\}$, but introduce a new 2-edge (from $e$). Therefore, we get the branching vector $(1,1-\alpha_2+\alpha_1)=(1,0.8)$, whose root is smaller than 2.168.}

\begin{figure}[!h]\centering
\frame{\includegraphics[scale=0.8]{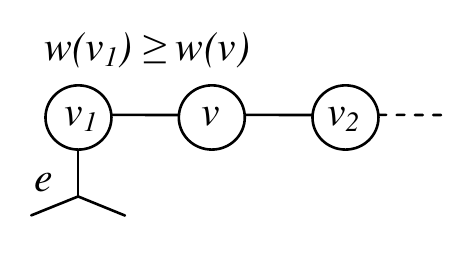}}
\caption{Rule \ref{rule:hs11} of \alg{WHS-Alg}.}
\end{figure}

\begin{branchrule}\label{rule:hs12}
{\normalfont [$G$ contains exactly two 2-edges, and there exist different $v_1,v,v_2\in V$ such that $\{v_1,v\},\{v,v_2\}\in E$, $|E(v_1)|=2$ and $w(v_1),w(v_2) < w(v)$]
\begin{enumerate}
\item If the result of \alg{WHS-Alg}($G'=(V\setminus \{v_1,v_2\}, E\setminus (E(v_1)\cup E(v_2))),w,W\!-w(\{v_1,v_2\}),k-2$) is not NIL: Return it along with $v_1$.
\item Else: Return \alg{WHS-Alg}($G'[V\setminus \{v_1,v\}],w,W\!-w(v),k-1)\cup\{v\}$, where $G'=(V\setminus \{v\}, E\setminus E(v))$.
\end{enumerate}}
\end{branchrule}

{\noindent Choosing $v_1$ (in the first branch), we need not choose $v$, since we can choose $v_2$ instead and get a hitting set that is not heavier than the one with $v$. In this case, we delete $\{v_1,v\}$ and $\{v,v_2\}$. Deleting $v_1$ (and therefore choosing $v$), we delete $\{v_1,v\}$ and $\{v,v_2\}$, but introduce a new 2-edge (from a 3-edge previously adjacent to $v_1$). Thus, the branching vector is $(2-\alpha_2,1-\alpha_2+\alpha_1)=(1.45,0.8)$, whose root is smaller than 2.168.}

\begin{figure}[!h]\centering
\frame{\includegraphics[scale=0.8]{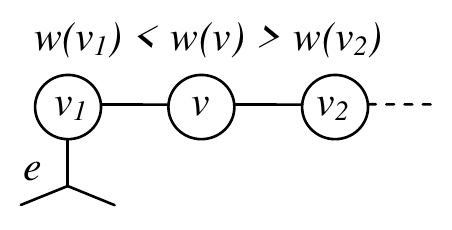}}
\caption{Rule \ref{rule:hs12} of \alg{WHS-Alg}.}
\end{figure}

\begin{branchrule}\label{rule:hs13}
{\normalfont [$G$ contains exactly two 2-edges, $\{v_1,v_2\}$ and $\{u_1,u_2\}$, where $v_1,v_2,u_1$ and $u_2$ are different vertices in $V$, and $|E(v_1)|\geq 2$]
\begin{enumerate}
\item If the result of \alg{WHS-Alg}($G'\!=\!(V\!\setminus\! \{v_1\}, E\!\setminus\! E(v_1)),w,W\!\!-\!w(v_1),k\!-\!1$) is not NIL: Return it along~with~$v_1$.
\item Else: Return \alg{WHS-Alg}($G'[V\!\setminus\! \{v_1,v_2\}],w,W\!\!-\!w(v_2),k\!-\!1)\cup\{v_2\}$, where $G'\!=\!(V\!\setminus\! \{v_2\}, E\!\setminus\! E(v_2))$.
\end{enumerate}}
\end{branchrule}

{\noindent This branching is exhaustive (we either choose $v_1$, or delete $v_1$ and thus choose $v_2$). Choosing $v_1$, we delete the 2-edge $\{v_1,v_2\}$. Deleting $v_1$ (and choosing $v_2$), we also delete the 2-edge $\{v_1,v_2\}$, but introduce a new 2-edge (from a 3-edge previously adjacent to $v_1$). We get a branching vector at least as good as $(1-\alpha_2+\alpha_1,1)=(0.8,1)$, whose root is smaller than 2.168.}

\begin{figure}[!h]\centering
\frame{\includegraphics[scale=0.8]{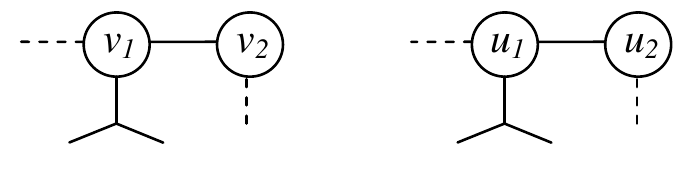}}
\caption{Rule \ref{rule:hs13} of \alg{WHS-Alg}. We illustrate, w.l.o.g, $u_1$ as a vertex of degree at least 2 (by Rule \ref{rule:whsdominate}, at least one vertex in $\{u_1,u_2\}$ has degree at least 2).}
\end{figure}

\bigskip

{\noindent From now on, since Rules \ref{rule:hs9}--\ref{rule:hs13} did not apply, $G$ contains exactly three 2-edges where there is no vertex that is contained in all of them, or at least four 2-edges.}

\begin{branchrule}\label{rule:hs14}
{\normalfont [$G$ contains exactly three 2-edges, and there exist different $v,u,r\in V$ such that $\{v,u\},\{v,r\}\in E$, and $|E(v)|\geq 3$]
\begin{enumerate}
\item If the result of \alg{WHS-Alg}($G'\!=\!(V\!\setminus \!\{v\}, E\!\setminus\! E(v)),w,W\!\!-\!w(v),k\!-\!1$) is not NIL: Return it along with $v$.
\item Else: Return \alg{WHS-Alg}($G'[V\setminus \{v,u,r\}],w,W\!-w(\{u,r\}),k-2)\cup\{u,r\}$, where $G'=(V\setminus \{u,r\}, E\setminus (E(u)\cup E(r)))$.
\end{enumerate}}
\end{branchrule}

{\noindent This branching is exhaustive (we either choose $v$, or delete $v$ and thus choose $u$ and $r$). Choosing $v$, we delete two 2-edges. Deleting $v$ (and choosing $u$ and $r$), we may also delete all 2-edges, but introduce a new 2-edge (from a 3-edge previously adjacent to $v$). We get a branching vector at least as good as $(1-\alpha_3+\alpha_1,2-\alpha_3+\alpha_1)=(0.55,1.55)$, whose root is smaller than 2.168.}

\begin{figure}[!h]\centering
\frame{\includegraphics[scale=0.8]{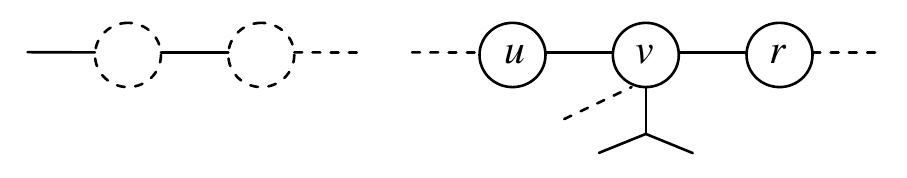}}
\caption{Rule \ref{rule:hs14} of \alg{WHS-Alg}. Two vertices are illustrated in dashed circles since each of them may be equal to $u$ or $r$.}
\end{figure}

\begin{reducerule}\label{rule:whsleafreduce}
{\normalfont [$G$ contains different $v,u,r\in V$ such that $\{v,u\},\{u,r\}\in E$ and $|E(v)|=1$]
Let $w'$ be $w$, except for $w'(r)=w(r)-(w(u)-w(v))$. 
\begin{enumerate}
\item If $w'(r)\leq 0$: Return \alg{WHS-Alg}($G[V\!\setminus\! \{v,u,r\}],w',W\!-\!w(v)\!-\!w(r),k\!-\!1)\cup\{v,r\}$.
\item Else: Return \alg{WHS-Alg}($G[V\setminus \{v,u\}],w',W\!-w(u),k-1$), along with $v$ if $r$ is in the returned result, and else along with $u$.
\end{enumerate}}
\end{reducerule}

{\noindent The correctness of this rule follows from the same arguments given for Cases 2(b) and 2(c) in Rule \ref{rule:deg1} of \alg{WVC-Alg} (see Section \ref{section:wvc1}). Note that, since $k$ is decreased by 1, $m(G,U)$ does not increase.}

\begin{figure}[!h]\centering
\frame{\includegraphics[scale=0.8]{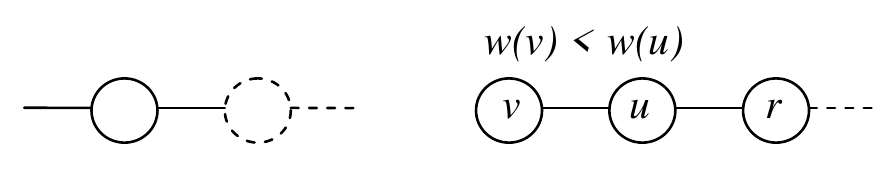}}
\caption{Rule \ref{rule:whsleafreduce} of \alg{WHS-Alg}. One vertex is illustrated in a dashed circle since it may be equal to $r$. Note that, by Rule \ref{rule:whsdominate}, $w(v)<w(u)$, and by Rule \ref{rule:hs14}, $|E(u)|=2$.}
\end{figure}

\begin{branchrule}\label{rule:hs16}
{\normalfont [$G$ contains exactly three 2-edges, and there exist different $v_1,v_2,v_3,v_4\in V$ such that $\{v_1,v_2\},\{v_2,v_3\},\{v_3,v_4\}\in E$ and $w(v_2)\geq w(v_3)$]
\begin{enumerate}
\item If the result of \alg{WHS-Alg}($G'\!=\!(V\!\setminus\! \{v_1\}, E\!\setminus\! E(v_1)),w,W\!\!-\!w(v_1),k\!-\!1$) is not NIL: Return it along~with~$v_1$.
\item Else: Return \alg{WHS-Alg}($G'[V\!\setminus\! \{v_1,v_2\}],w,W\!\!-\!w(v_2),k\!-\!1)\cup\{v_2\}$, where $G'\!=\!(V\!\setminus\! \{v_2\}, E\!\setminus\! E(v_2))$.
\end{enumerate}}
\end{branchrule}

{\noindent This branching is exhaustive (we either choose $v_1$, or delete $v_1$ and thus choose $v_2$). Choosing $v_1$, we delete one 2-edge, and then apply the first case of Rule \ref{rule:whsdominate}. Deleting $v_1$, we delete two 2-edges, but introduce a new 2-edge (from a 3-edge previously adjacent to $v_1$). Therefore, we get a branching vector that is at least as good as $(2-\alpha_3,1-\alpha_3+\alpha_2)=(1.2,0.75)$, whose root is smaller than 2.168.}

\begin{figure}[!h]\centering
\frame{\includegraphics[scale=0.8]{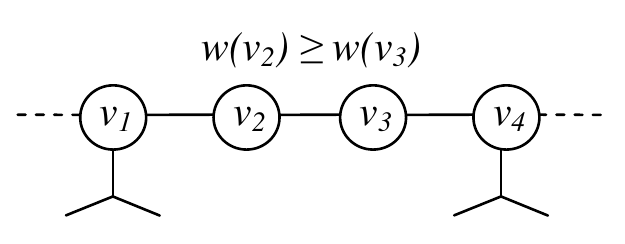}}
\caption{Rule \ref{rule:hs16} of \alg{WHS-Alg}. Note that, by Rule \ref{rule:hs14}, $|E(v_2)|=|E(v_3)|=2$, and by Rule \ref{rule:whsleafreduce}, $|E(v_1)|,|E(v_4)|\geq 2$.}
\end{figure}

\begin{branchrule}\label{rule:hs17}
{\normalfont [$G$ contains exactly three 2-edges, and there exist $v,u\in V$ such that $\{v,u\}\in E$, the only 2-edge in $E(v)\cup E(u)$ is $\{v,u\}$ and $|E(v)|\geq 3$]
\begin{enumerate}
\item If the result of \alg{WHS-Alg}($G'\!=\!(V\!\setminus\! \{v\}, E\!\setminus \!E(v)),w,W\!\!-\!w(v),k\!-\!1$) is not NIL: Return it along with $v$.
\item Else: Return \alg{WHS-Alg}($G'[V\!\setminus\!\{v,u\}],w,W\!\!-\!w(u),k\!-\!1)\cup\{u\}$, where $G' \!=\! G(V\!\setminus\! \{u\}, E\!\setminus\! E(u))$.
\end{enumerate}}
\end{branchrule}

{\noindent This branching is exhaustive. Choosing $v$, we delete the 2-edge $\{v,u\}$. Choosing $u$, we also delete this edge, but introduce two new 2-edges (from 3-edges previously adjacent to $v$). We get a branching vector that is at least as good as $(1-\alpha_3+\alpha_2,1-\alpha_3+\alpha_4)=(0.75,1.07)$, whose root is smaller than 2.168.}

\begin{figure}[!h]\centering
\frame{\includegraphics[scale=0.8]{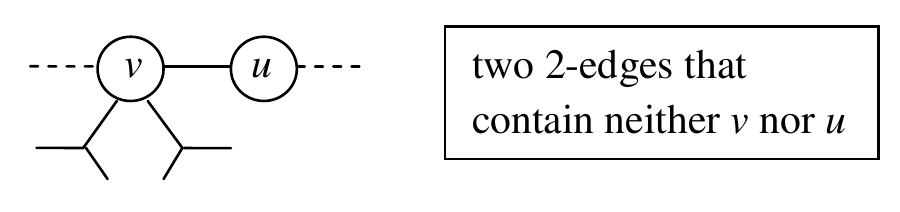}}
\caption{Rule \ref{rule:hs17} of \alg{WHS-Alg}.}
\end{figure}

\begin{branchrule}\label{rule:hs18}
{\normalfont [$G$ contains exactly three 2-edges, and there are $v,u\in V$ such that $\{v,u\}\in E$, the only 2-edge in $E(v)\cup E(u)$ is $\{v,u\}$, and $w(v)\geq w(u)$] Let $r_1$ and $r_2$ be the vertices, excluding $v$, in the 3-edge $e\in E(v)\setminus\{\{v,u\}\}$. Let $S$ be the set of vertices in 1-edges in $\widetilde{G}[V\setminus e]$, where $\widetilde{G}=(V\setminus \{v\}, E\setminus E(v))$.
\begin{enumerate}
\item If the result of \alg{WHS-Alg}($G'[V\setminus (e\!\cup\! S)],w,W\!-\!w(\{v\}\cup S),k\!-\!1\!-\!|S|$) is not NIL, where $G'=(V\setminus (\{v\}\!\cup\! S), E\setminus E(\{v\}\!\cup\! S))$: Return it along with $\{v\}\!\cup\! S$.
\smallskip
\item Else: Return \alg{WHS-Alg}($G'[V\setminus \{v,u\}],w,W\!-w(u),k-1)\cup\{u\}$, where $G'=(V\setminus \{u\}, E\setminus E(u))$.
\end{enumerate}}
\end{branchrule}

{\noindent As in the proof of Rule \ref{rule:12edgeb}, we can assume w.l.o.g that $|E(r_1)|\geq 2$. Choosing $v$, we need to choose neither $r_1$ nor $r_2$, since then we can replace $v$ by $u$ and get a hitting set that is not heavier than the one with $v$ (the choice of $u$ is examined in the second branch). Thus, in this branch, we delete $r_1$ and $r_2$, and thus we must choose the vertices in $S$. For the analysis of the branching vector, note that, in this branch, we delete $\{v,u\}$. If $S=\emptyset$, we introduce a new 2-edge (from a 3-edge previously adjacent to $r_1$), such that we get three edges that do not all have a common vertex (if we got exactly three 3-edges that have a common vertex, Rule \ref{rule:hs14} would have been applicable, and we would not have reached this rule). If $S\neq \emptyset$, at the worst case, $|S|=1$ and we delete the other two 2-edges (overall, we then delete all the 2-edges in $E$). Deleting $v$ (and choosing $u$) in the second branch, we delete $\{v,u\}$, but introduce a new 2-edge, $\{r_1,r_2\}$, such that we get three 2-edges that do not all have a common vertex (again, if this is not the case, Rule \ref{rule:hs14} should have been applied). Thus, we get a branching vector that is at least as good as $\max\{(1,1),(2-\alpha_3,1)\} = (1,1)$, whose root is 2.}

\begin{figure}[!h]\centering
\frame{\includegraphics[scale=0.8]{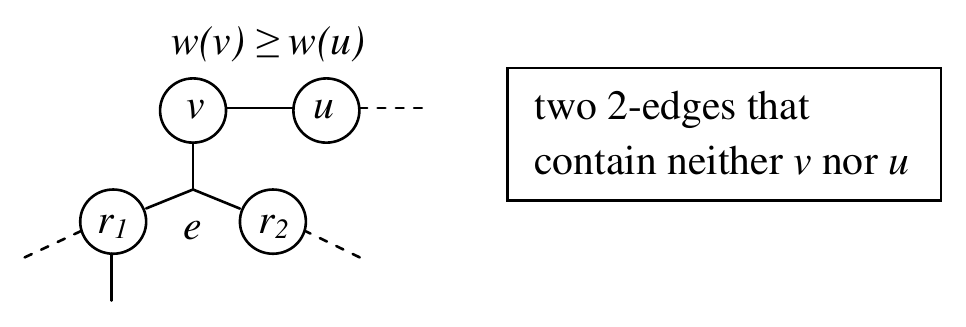}}
\caption{Rule \ref{rule:hs18} of \alg{WHS-Alg}. Note that, since Rules \ref{rule:whsdominate} and \ref{rule:hs17} did not apply, $|E(v)|=2$.}
\end{figure}

{\noindent From now on, since Rules \ref{rule:hs14}--\ref{rule:hs18} did not apply, $G$ contains at least four 2-edges.}

\begin{branchrule}\label{rule:hs19}
{\normalfont [$G$ contains at least four 2-edges that have a common vertex $v$] Let $S$ be set of vertices, excluding $v$, of the 2-edges in $E(v)$.
\begin{enumerate}
\item If the result of \alg{WHS-Alg}($G'\!=\!(V\!\setminus\! \{v\}, E\!\setminus\! E(v)),w,W\!\!-\!w(v),k\!-\!1$) is not NIL: Return it along with $v$.
\item Else: Return \alg{WHS-Alg}($G'[V\setminus (S\!\cup\!\{v\})],w,W\!\!-\!w(S),k\!-\!|S|)\cup S$, where $G'\!=\!(V\!\setminus\! S, E\!\setminus\! E(S))$.
\end{enumerate}}
\end{branchrule}

{\noindent This branching is exhaustive (we either choose $v$, or delete $v$ and thus choose $S$). We get a branching vector that is at least as good as $(1-\alpha_4,4-\alpha_4)=(0.13,3.13)$, whose root is smaller than 2.168.}

\begin{branchrule}\label{rule:hs20}
{\normalfont [$G$ contains three 2-edges that have a common vertex $v$] Let $S$ be set of vertices, excluding $v$, of the 2-edges in $E(v)$.
\begin{enumerate}
\item If the result of \alg{WHS-Alg}($G'\!=\!(V\!\setminus\! \{v\}, E\!\setminus\! E(v)),w,W\!\!-\!w(v),k\!-\!1$) is not NIL: Return it along with $v$.
\item Else: Return \alg{WHS-Alg}($G'[V\setminus (S\!\cup\!\{v\})],w,W\!\!-\!w(S),k\!-\!|S|)\cup S$, where $G'\!=\!(V\!\setminus\! S, E\!\setminus\! E(S))$.
\end{enumerate}}
\end{branchrule}

{\noindent This branching is exhaustive (we either choose $v$, or delete $v$ and thus choose $S$). Note that $|S|=3$, and that in the first branch we do not delete the 2-edge that is not adjacent to $v$. Thus, we get a branching vector that is at least as good as $(1-\alpha_4+\alpha_1,3-\alpha_4)=(0.48,2.13)$, whose root is smaller than 2.168.}

\begin{branchrule}\label{rule:hs21}
{\normalfont [There are different $v,u,r\in V$ such that $\{v,u\},\{u,r\}\in E$ and $|E(u)|\geq 3$]
\begin{enumerate}
\item If the result of \alg{WHS-Alg}($G'\!=\!(V\!\setminus\! \{u\}, E\!\setminus\! E(u)),w,W\!\!-\!w(u),k\!-\!1$) is not NIL: Return it along~with~$u$.
\item Else: Return \alg{WHS-Alg}($G'[V\setminus \{v,u,r\}],w,W\!-w(\{v,r\}),k-2)\cup\{v,r\}$, where $G'=(V\setminus \{v,r\}, E\setminus (E(v)\cup E(r)))$.
\end{enumerate}}
\end{branchrule}

{\noindent This branching is exhaustive (we either choose $u$, or delete $u$ and thus choose $v$ and $r$). Since the two previous rules did not apply, there are at least two 2-edges that are not adjacent to $u$; therefore, in the first branch, we do not delete at least two 2-edges. In the second branch, we may delete all existing 2-edges, but introduce a new one (since $|E(u)|\geq 3$). Thus, we get a branching vector that is at least as good as $(1-\alpha_4+\alpha_2,2-\alpha_4+\alpha_1)=(0.68,1.48)$, whose root is smaller than 2.168.}

\begin{branchrule}\label{rule:hs22}
{\normalfont [Remaining case]
Let $v$ be a vertex in $V$ that is adjacent to exactly one 2-edge, $\{v,u\}$, for some $u\in V$, where $|E(v)|\geq 2$.
\begin{enumerate}
\item If the result of \alg{WHS-Alg}($G'\!=\!(V\!\setminus\! \{v\}, E\!\setminus\! E(v)),w,W\!\!-\!w(v),k\!-\!1$) is not NIL: Return it along with $v$.
\item Else: Return \alg{WHS-Alg}($G'[V\setminus \{v,u\}],w,W\!-w(u),k-1)\cup S$, where $G'=(V\setminus \{u\}, E\setminus E(u))$.
\end{enumerate}}
\end{branchrule}

{\noindent Since the previous rule did not apply, the rule is well-defined (i.e., there exists $v$ as defined in this rule), and $u$ is contained in at most one 2-edge that is not $\{v,u\}$. The branching is exhaustive (we either choose $v$, or delete $v$ and thus choose $u$), and thus correct. In the first branch, we delete only one 2-edge. Thus, there remain at least three 2-edges, where, if there remain exactly three, they do not have a common vertex (otherwise Rule \ref{rule:hs19} or \ref{rule:hs20} was applied). In the second branch, we delete, at worst, two 2-edges. However, we also introduce a new 2-edge (from a 3-edge adjacent to $v$), and thus have at least three 2-edges, where, if there are exactly three, they do not have a common vertex (otherwise the previous rule was applied). Thus, we get a branching vector that is at least as good as $(1\!-\!\alpha_4\!+\!\alpha_3,1\!-\!\alpha_4\!+\!\alpha_3)=(0.93,0.93)$, whose root is smaller than~2.168.}

\mysubsection{A Faster Algorithm for Graphs that Have a Small HS}\label{section:whs2}

In this appendix, we show that our algorithms for {\sc $k$-WVC} (given in Section \ref{section:wvc1} and Appendix \ref{section:wvc2}) can be used to develop an algorithm for {\sc W3HS} that is fast on graphs that have a small (unweighted) hitting set. More precisely, we prove the following.

\begin{theorem}\label{theorem:whs2}
{\sc W3HS} can be solved in $O^*(1.381^{s-t}2.381^t)$ time and polynomial space, or $O^*(1.363^{s-t}2.363^t)$ time and $O^*(1.363^s)$ space.
\end{theorem}

For example, if $s=2t$, we obtain algorithms which use only $O^*(1.814^s)$ time and polynomial space, and $O^*(1.795^s)$ time and $O^*(1.363^s)$ space, respectively. Next, we develop \alg{WHS*-Alg}, an algorithm for which we prove the following result, which implies the correctness of Theorem \ref{theorem:whs2}.

\begin{lem}
\alg{WHS*-Alg} solves {\sc $k$-W3HS} in $O^*(1.381^{k-t}2.381^t)$ time and polynomial space, or $O^*(1.363^{k-t}2.363^t)$ time and $O^*(1.363^k)$ space.
\end{lem}

\begin{proof}
Let \alg{ALG}($G,w,W,k$) be an algorithm for {\sc $k$-WVC}.
On a high level, building upon the approach of \cite{hs2010a}, \alg{WHS*-Alg} first computes a minimum(-size) hitting set $U$. Then, \alg{WHS*-Alg} considers every subset $U'$ of $U$ as a possible ``partial solution'', which should be completed to a solution (i.e., a light hitting set) by adding vertices from $V\setminus U$ (only). \alg{WHS*-Alg} attempts to complete $U'$ in this manner by running \alg{ALG} on a certain subgraph of $G[V\setminus U]$. We now give the pseudocode of \alg{WHS*-Alg}, which contains the exact definition of this subgraph.

\begin{algorithm}[!ht]
\caption{\alg{WHS*-Alg}($G=(V,E),w: V\rightarrow \mathbb{R}^{\geq 1},W,k$)}
\begin{algorithmic}[1]
\STATE\label{step:minhs} Compute a minimum(-size) hitting set $U$ for $G$ using the algorithm of \cite{hs2007}.
\FORALL{$U'\subseteq U$}\label{step:loop}
	\STATE Let $S$ include every $v\in V$ for which there exists $e\in E$ such that $e\setminus\{v\}\subseteq (U\setminus U')$.
	\IF{$S\cap (U\setminus U') \neq \emptyset$}\label{step:skipite}
		\STATE Skip the rest of the current iteration.
	\ENDIF
	\STATE Let $G_{U'}$ be defined by $V(G_{U'})\!=\!V\setminus (U'\!\cup\! S)$ and $E(G_{U'})\!=\!E\setminus E(U'\!\cup\! S)$.
	\STATE Let $S'\Leftarrow$ \alg{ALG}($G_{U'}[V\setminus(U\cup S)], w, W-w(U'\cup S), k-|U'\cup S|$).	
	\IF{$S'\neq$ NIL}
		\STATE Return $U'\cup S\cup S'$.
	\ENDIF
\ENDFOR
\STATE Return NIL.
\end{algorithmic}
\end{algorithm}

\myparagraph{Correctness} Denote $\widetilde{G}_{U'} = G_{U'}[V\setminus(U\cup S)]$. First, observe that since $U$ is a hitting set, $\widetilde{G}_{U'}$ does not contain a 3-edge, and thus it is a legal input for \alg{ALG}.

Let $U'\cup S\cup S'$ be a solution returned by \alg{WHS*-Alg}. Clearly, $w(U'\cup S\cup S')\leq W$. Consider some edge $e\in E$. By Step \ref{step:skipite}, $e$ is not a subset of $U\setminus U'$. Thus, if $e\cap (U'\cup S) = \emptyset$, there is an edge $e'\in \widetilde{G}_{U'}$ such that $e'\subset e$. Since $S'$ is a vertex cover for $\widetilde{G}_{U'}$, we get that $e'\cap S'\neq\emptyset$. Thus, $U'\cup S\cup S'$ is a hitting set for $G$.

Now, let $A$ be a hitting set for $G$ such that $w(A)\leq W$ and $|A|\leq k$. Denote $U'=A\cap U$. Suppose that \alg{WHS*-Alg} reached the iteration of Step \ref{step:loop} that corresponds to $U'$ (by the former direction, if the algorithm has not reached this iteration, then it has already returned a correct solution). Since $S\subseteq A$, and $A\setminus(U'\cup S)$ is a vertex cover for $\widetilde{G}_{U'}$ of weight at most $W-(U'\cup S)$ and size at most $k-|U'\cup S|$, \alg{WHS*-Alg} returns a solution in the current iteration, which is correct (by the former direction).

\myparagraph{Time and Space Complexities} The algorithm of \cite{hs2007}, called in Step \ref{step:minhs}, uses $O^*(2.076^t)$ time and polynomial space. By Section \ref{section:wvc1}, we can choose \alg{ALG} such that \alg{WHS*-Alg} uses polynomial space and runs in time bounded by

\[
\displaystyle{O^*(\sum_{U'\subseteq U}1.381^{k-|U'|}) = O^*(1.381^{k-t}\sum_{U'\subseteq U}1.381^{t-|U'|}) = O^*(1.381^{k-t}2.381^t)}
\]

Moreover, by Appendix \ref{section:wvc2}, we can choose \alg{ALG} such that \alg{WHS*-Alg} uses $O^*(1.363^k)$ space and runs in time bounded by

\[
\displaystyle{O^*(\sum_{U'\subseteq U}1.363^{k-|U'|}) = O^*(1.363^{k-t}\sum_{U'\subseteq U}1.363^{t-|U'|}) = O^*(1.363^{k-t}2.363^t)}
\]\qed
\end{proof}
\section{Algorithms for Weighted Edge Dominating Set}\label{sec:WEDS}

In this appendix, we develop FPT algorithms for {\sc WEDS}. First, in Appendix \ref{section:edsvc}, we recall known results on the relation between edge dominating sets and vertex covers. We also observe (in Appendix \ref{section:edsvc}) that, relying on the flexible use of the parameter $k$ in our framework, the algorithm for {\sc EDS} of \cite{eds2013} can be modified to solve {\sc $k$-WEDS} in $O^*(2.315^k)$ time and polynomial space. Thus, {\sc WEDS} can be solved in $O^*(2.315^s)$ time and polynomial space. We complement this result (in Appendix~\ref{section:weds2}) by developing an $O^*(3^t)$ time and polynomial space algorithm for {\sc WEDS}.

\subsection{A Relation Between Edge Dominating Sets and Vertex Covers}\label{section:edsvc}

Consider some instance $(G=(V,E), w: V\rightarrow \mathbb{R}^{\geq 1}, W, k)$ of {\sc WEDS}. Let $U\subseteq E$ be an edge dominating set of $G$, and let $V(U)=\bigcup_{e\in U}e$ denote
 the set of endpoints of the edges in $U$. As observed in \cite{weds2006}, $V(U)$ is a vertex cover of $G$. Indeed, if there is an edge that is not covered by a vertex in $V(U)$, then this edge is not covered by an edge in $U$, which contradicts the fact that $U$ is an edge dominating set. Thus, we have the following observation.

\begin{obs}\label{obs:edstovc}
For any edge dominating set $U$ of $G$, $V(U)$ is a vertex cover (of size at most $2|U|$) of $G$.
\end{obs}

We say that a subset of vertices $A\subseteq V$ {\em represents} another subset of vertices $B\subseteq V$ if $A\subseteq B$. 
Also, a family ${\cal A}$ of subsets of vertices is an {\em $\ell$-representation of vertex covers}, for some $\ell\in\mathbb{N}$, if for any vertex cover $B$ of $G$ of size at most $\ell$, for which there exists an edge dominating set $U$ of size at most $k$ that satisfies $U\subseteq V(B)$, there exists $A\in{\cal A}$ that represents $B$. Now, if for every $A\in{\cal A}$, $G[V\setminus A]$ contains only connected components that are cliques on at most three vertices (i.e., isolated vertices, paths on two vertices and triangles), we further say that ${\cal A}$ is a {\em good} $\ell$-representation of vertex covers. For such representations, we can use the following result, whose proof is given in \cite{exactweds} (generalizing a~result~of~\cite{weds2006}).

\begin{lem}\label{lemma:wedsgoodrep}
Let $A$ be a vertex cover that belongs to some good $\ell$-representation of vertex covers, and let ${\cal U}$ be the set of every edge dominating set $U$ of $G$ such that $A\subseteq V(U)$. Then, one can compute in polynomial time an edge dominating set of $G$ that has minimum weight among those in ${\cal U}$.
\end{lem}

Thus, by Observation \ref{obs:edstovc} and Lemma \ref{lemma:wedsgoodrep}, we have the following proposition, used to develop algorithms for {\sc WEDS}.

\begin{prop}\label{prop:weds}
An edge dominating set $U$ of $G$ whose weight is smaller or equal to the weight of any edge dominating set of $G$ whose size is at most $k$ can be computed as follows. For an arbitrary good $2k$-representation of vertex covers $\cal A$ of $G$, iterate over every set $A\in{\cal A}$, and compute a corresponding best edge dominating set using Lemma \ref{lemma:wedsgoodrep}. Then, return an edge dominating set of minimum weight among the computed ones.
\end{prop}

In particular, note that the returned edge dominating set may be of size larger than $k$, which complies with our flexible use of the parameter $k$ in the definition of {\sc $k$-WEDS}. Now, the algorithm of \cite{eds2013} iterates over a family of sets that is {\em almost} a $2k$-representation of vertex covers. It is straightforward to modify this algorithm to iterate over sets of a $2k$-representation of vertex covers without increasing its time and space complexities.\footnote{The modification simply involves branching on vertices of cliques that contain at least four vertices, which are ignored by the algorithm of \cite{eds2013}, since for {\sc EDS} one can use a more relaxed definition for a $2k$-representation of vertex covers (as shown~in~\cite{exactweds}).} We thus obtain the following results.

\begin{theorem}
{\sc $k$-WEDS} can be solved in $O^*(2.315^k)$ time and polynomial space.
\end{theorem}

\begin{cor}
{\sc WEDS} can be solved in $O^*(2.315^s)$ time and polynomial space.
\end{cor}

For our second result, given in the following appendix, it is enough to consider the following weaker proposition.

\begin{prop}\label{prop:weds2}
Let $MinVC$ be the set of all minimal vertex covers of $G$.
An edge dominating set $U$ of $G$ can be computed as follows. For an arbitrary superset ${\cal A}$ of $MinVC$, iterate over every vertex cover $A\in{\cal A}$, and compute a corresponding best edge dominating set using Lemma \ref{lemma:wedsgoodrep}. Then, return an edge dominating set of minimum weight among the computed ones.
\end{prop}

\subsection{{\sc WEDS} Parameterized by the Size of a Minimum EDS}\label{section:weds2}

In this appendix, we develop an algorithm, \alg{WEDS*-Alg}, that solves the following problem.

\myparagraph{WEDS*} Given an instance of {\sc WEDS}, along with a minimum(-size) edge dominating set $U$, return a superset of $MinVC$.

\smallskip

For this algorithm, we obtain the following result.

\begin{theorem}\label{theorem:weds*}
\alg{WEDS*-Alg} solves {\sc WEDS*} in time $O^*(3^t)$ and polynomial space.
\end{theorem}

Given an instance $(G,w,W)$ of {\sc WEDS}, we can find, in time $O^*(2.315^t)$ and polynomial space, a minimum edge dominating set $U$, using the algorithm of \cite{eds2013}. Thus, by Proposition \ref{prop:weds2}, we obtain the following result.

\begin{cor}\label{cor:wedsT}
{\sc WEDS} can be solved in time $O^*(3^t)$ and polynomial space.
\end{cor}

We now turn to present \alg{WEDS*-Alg}, whose pseudocode is given below. This algorithm computes a superset $\cal A$ of $MinVC$ by adding (to $\cal A$) every set that consists of a subset $S\subseteq V(U)$ that covers all the edges in $U$, along with the neighbor set of $V(U)\setminus S$.

\begin{algorithm}[!ht]
\caption{\alg{WEDS*-Alg}($G=(V,E),w: V\rightarrow \mathbb{R}^{\geq 1},W,U$)}
\begin{algorithmic}[1]
\STATE Initialize ${\cal A}\Leftarrow\emptyset$.
\FORALL{$S\subseteq V(U)$ such that $S$ contains at least one endpoint of any edge in $U$}\label{step*:loop}
	\STATE Let $A_{S}\Leftarrow S\cup N(V(U)\setminus S)$.
	\STATE Add $A_{S}$ to $\cal A$.
\ENDFOR
\STATE Return $\cal A$.
\end{algorithmic}
\end{algorithm}

{\noindent First, note that iterating over every subset $S$ in Step \ref{step*:loop} can be perfomed by choosing, for every edge $\{v,u\}\in U$, (a) $v$, (b) $u$, or (c) $v$ and $u$ (i.e., there are three choices per edge). Therefore, \alg{WEDS*-Alg} runs in time $O^*(3^t)$, and it clearly uses polynomial space. To prove its correctness, we need to show that it returns a superset $\cal A$ of $MinVC$. To this end, consider an arbitrary minimal vertex cover $X\in MinVC$, and denote $S=X\cap V(U)$. Since $X$ is a vertex cover, there is an iteration of Step \ref{step*:loop} that corresponds to this set $S$, and we can let $A_S$ be the corresponding set that \alg{WEDS*-Alg} adds to $\cal A$. We now show that $X = A_S$, which concludes the proof of Theorem \ref{theorem:weds*}. Since $X$ is a vertex cover, it contains $N(V(U)\setminus S)$. Thus, $A_S\subseteq X$. Suppose, by way of contradiction, that there exists a vertex $v$ in $X\setminus A_S$. Since $X$ is a {\em minimal} vertex cover, there is a vertex $u$ in $N(v)\setminus X$. Therefore, $u\notin A_S$. However, since $U$ is an edge dominating set, it contains an edge $e$ such that one of its endpoints belongs to $\{v,u\}$ (otherwise the edge $\{v,u\}$ is not covered). By the definition of $A_S$, for every vertex $r\in V(U)$, it contains either $r$ or all of its neighbors (or both), and thus we have a contradition. We get that, indeed, $X = A_S$.}
\section{Weighted Max Internal Out-Branching}\label{section:wiob}

In this appendix, we aim to demonstrate in a simple manner that our framework is useful in solving weighted maximization problems, or problems for which existing algorithms are based on techniques other than bounded search trees. To this end, we present an FPT algorithm for {\sc WIOB} which uses time $O^*(6.855^s)$, or randomized time $O^*(4^sW)$. Our flexible use of the parameter $k$ (which defines {\sc $k$-WIOB}) allows us to easily obtain these algorithms, by using FPT algorithms for the related {\sc Weighted $k$-ITree} problem, defined below. Note that the algorithm of \cite{corrrepresentative} relies on the combinatorial representative sets technique \cite{representative}, and the algorithm of \cite{thesis11,ipec13} uses the algebraic multilinear detection technique \cite{multilineardetection,appmultilinear}.

The problem {\sc Weighted $k$-ITree} is defined as follows.

\myparagraph{Weighted $k$-ITree} Given a directed graph $G=(V,E)$, a vertex $r\in V$, a weight function $w: V\rightarrow \mathbb{R}^{\geq 1}$, a weight $W\in\mathbb{R}^{\geq 1}$, and a parameter $k\in\mathbb{N}$, find an out-tree of $G$ (i.e., a subtree of $G$ having exactly one vertex of in-degree 0, called the root) that is rooted at $r$, and contains exactly $k$ internal vertices and at most $k$ leaves, such that the total weight of its internal vertices is at least $W$ (if one exists).

\medskip

The following result, implicitly given in \cite{kIOB49k} to show a connection between {\sc IOB} and {\sc $k$-ITree}, holds also for the weighted variants of these problems.

\begin{lem}\label{lemma:reduceIOB}
Let $G=(V,E)$ be a directed graph. Also, let $r$ be a vertex in $V$ such that $G$ has an out-branching $H$ rooted at $r$.
\begin{itemize}
\item The graph $G$ has an out-tree rooted at $r$, having the same set $S$ of internal vertices as $H$, and at most $|S|$ leaves.
\item Given an out-tree $T$ of $G$ rooted at $r$, one can construct in polynomial time an out-branching of $G$ whose set of internal vertices contains the set of internal vertices of $T$.
\end{itemize}
\end{lem}

The algorithm given in \cite{corrrepresentative} for {\sc $k$-ITree} is based on the representative families technique \cite{representative}. This algorithm can be modified to solve {\sc Weighted $k$-ITree}, maintaining the same $O^*$ time and space complexities, by replacing its representative families computations by max representative families computations (see, e.g., \cite{representative} and \cite{mathpack14}). Another algorithm for {\sc $k$-ITree}, given in \cite{thesis11,ipec13}, is based on the multilinear detection technique \cite{multilineardetection,appmultilinear}. This algorithm can also be modified to solve {\sc Weighted $k$-ITree}, increasing its $O^*$ time and space complexities by a factor of $W$, by introducing a variable that tracks the weights of the solutions associated with the computed monomials (see, e.g., \cite{integerweights} and \cite{techreport1}). Thus, we obtain algorithms, called \alg{WITree-Alg} and \alg{RandWITree-Alg}, for which we have the following result.

\begin{lem}\label{lemma:treealgs}
\alg{WITree-Alg} solves {\sc Weighted $k$-ITree} in $O^*(6.855^k)$ time and space, and \alg{RandWITree-Alg} solves {\sc Weighted $k$-ITree} in $O^*(4^kW)$ randomized time and $O^*(W)$ space.
\end{lem}

Let \alg{ALG} be an algorithm for {\sc Weighted $k$-ITree} that uses $O^*(f(k,W))$ time and $O^*(g(k,W))$ space. We next present the pseudocode of \alg{WIOB-Alg}, an algorithm that solves {\sc $k$-WIOB} by using Lemma \ref{lemma:reduceIOB} and \alg{ALG}. Note that the condition in Step \ref{step:checkoutbranch} can be easily checked in linear time (e.g., using DFS).

\begin{algorithm}[!ht]
\caption{\alg{WIOB-Alg}($G=(V,E),w: V\rightarrow \mathbb{R}^{\geq 1},W,k$)}
\begin{algorithmic}[1]

\FORALL{$r\in V$}
\FOR{$k'=1,2,\ldots,k$}
	\IF{$G$ does not have an out-branching rooted at $r$}\label{step:checkoutbranch}
		\STATE Skip the rest of the current iteration.
	\ENDIF

	\STATE $T\Leftarrow $ \alg{ALG}$(G,r,w,W,k')$.
	
	\IF{$T\neq$ NIL}
		\STATE Extend $T$ to an out-branching $OB$ according to Lemma \ref{lemma:reduceIOB}.
		\STATE Return $OB$.
	\ENDIF
	
\ENDFOR	
\ENDFOR

\STATE Return NIL.
\end{algorithmic}
\end{algorithm}

The correctness of \alg{WIOB-Alg} clearly follows from the correctness of Lemma \ref{lemma:reduceIOB} and \alg{ALG}. Note that, in this context, our flexible use of the parameter $k$ (in the definition of {\sc $k$-WIOB}) is crucial, since \alg{WIOB-Alg} may return out-branchings that have more than $k$ internal vertices. Finally, note that \alg{WIOB-Alg} runs in $O^*(f(k,W))$ time and uses $O^*(g(k,W))$ space.

By the above arguments and Lemma \ref{lemma:treealgs}, we obtain the following result.

\begin{lem}
{\sc $k$-WIOB} can be solved in $O^*(6.855^k)$ time and space. Moreover, {\sc $k$-WIOB} can be solved in $O^*(4^kW)$ randomized time and $O^*(W)$ space.
\end{lem}

Thus, we have proved the correctness of the following theorem.

\begin{theorem}\label{theorem:wiob}
{\sc WIOB} can be solved in $O^*(6.855^s)$ time and space. Alternatively, it can be solved in $O^*(4^sW)$ randomized time and $O^*(W)$ space.
\end{theorem}

\section{The Bounded Search Trees Technique: An Example}\label{app:boundedsearch}

In the following, we give an example of a simple application of the bounded search trees technique. We start by presenting 
an algorithm, \alg{ALG1}, for {\sc VC}. Then, we explain how to solve {\sc WVC} in a similar manner.
 Finally, we demonstrate how to improve upon these algorithms by  integrating them in our framework (see Section 
\ref{section:technique}). The applications given below are far from optimal, and are only presented for the sake of
completeness and the clarity of the paper.

\subsection{An Algorithm for {\sc VC}}
Recall that, in solving {\sc VC}, we need to find a vertex cover of size at most $t$. Clearly, if we have an instance 
where $t < 0$, we need to choose a negative number of vertices (since $t < 0$), which is impossible. Therefore, 
we can return NIL. This leads to the following reduction rule.

\setcounter{reducerule}{0}
\begin{reducerule}
{\normalfont [$t<0$]
Return NIL.}
\end{reducerule}

{\noindent Since the algorithm always chooses the first applicable rule, in the following rules, we may assume that $t\geq 0$. If $E=\emptyset$, we need to choose a nonnegative number of vertices to cover zero edges.
 Clearly, the empty set is a valid solution. This leads to a second reduction rule:}

\begin{reducerule}
{\normalfont [$E=\emptyset$]
Return $\emptyset$.}
\end{reducerule}

{\noindent We may now assume that $t\geq 0$ and $E\neq\emptyset$. Suppose that the graph $G$ contains an edge $e=\{v,u\}$ such that one of its endpoints, $v$, is a leaf in $G$. To obtain a solution, we must choose at least one of the vertices $v$ and $u$ (to cover $\{v,u\}$). 
Moreover, adding $u$ to the solution is always ``at least as good'' as choosing $v$ to the solution. Indeed, given a vertex cover that contains $v$, we get a vertex cover of the same or smaller size by removing $v$ and
 inserting $u$ to the cover. Therefore, we can insert $u$ to the solution, 
and attempt to choose at most $t-1$ vertices that cover the remaining edges (i.e., the edges not covered by $u$) by recursively calling \alg{ALG1} with $G[V\setminus\{v,u\}]$ and $t-1$. This is the action performed by the following reduction rule.}

\begin{reducerule}
{\normalfont [There is a leaf $v$ in $G$] Let $u$ be the neighbor of $v$ in $G$.
Return \alg{ALG1}$(G[V\setminus\{v,u\}],t-1)$, along with $u$.\footnote{We assume that adding elements to NIL results in NIL.}}
\end{reducerule}

{\noindent Now, the graph $G$ contains at least one edge (by Rule 2), and does not contain leaves (by Rule 3). Therefore, $G$ contains at least one vertex, $v$, that has at least two neighbors. To cover the edges in $E(v)$, we can choose all of the neighbors of $v$. However, to this end, we can also choose $v$. In the latter choice, we add only one vertex to the solution, but cover only the edges in $E(v)$. Thus, we have two options, and accordingly, we perform a branching rule that consists of two branches:}

\begin{branchrule}
{\normalfont [Remaining case] Let $v$ be a vertex of maximum~degree~in~$G$.
\begin{enumerate}
\item If the result of \alg{ALG1}($G[V\setminus N(v)],t-|N(v)|$) is not NIL: Return it along with~$N(v)$.
\item Else: Return \alg{ALG1}($G[V\setminus \{v\}],t-1$), along with~$v$.
\end{enumerate}}
\end{branchrule}

{\noindent In this rule, the parameter $t$ is decreased by $|N(v)|\geq 2$ in the first branch, and by 1 in the second branch. Therefore, the branching vector associated with this rule is at least as good as (2,1), whose root is smaller than 1.619. Overall, 
as described in Section \ref{section:preliminaries}, this implies that \alg{ALG1} runs in time $O^*(1.619^t)$. An illustration of a search tree corresponding to this algorithm is given in Fig.~\ref{fig:tree}.}

\begin{figure}[!h]\centering
\frame{\includegraphics[scale=0.75]{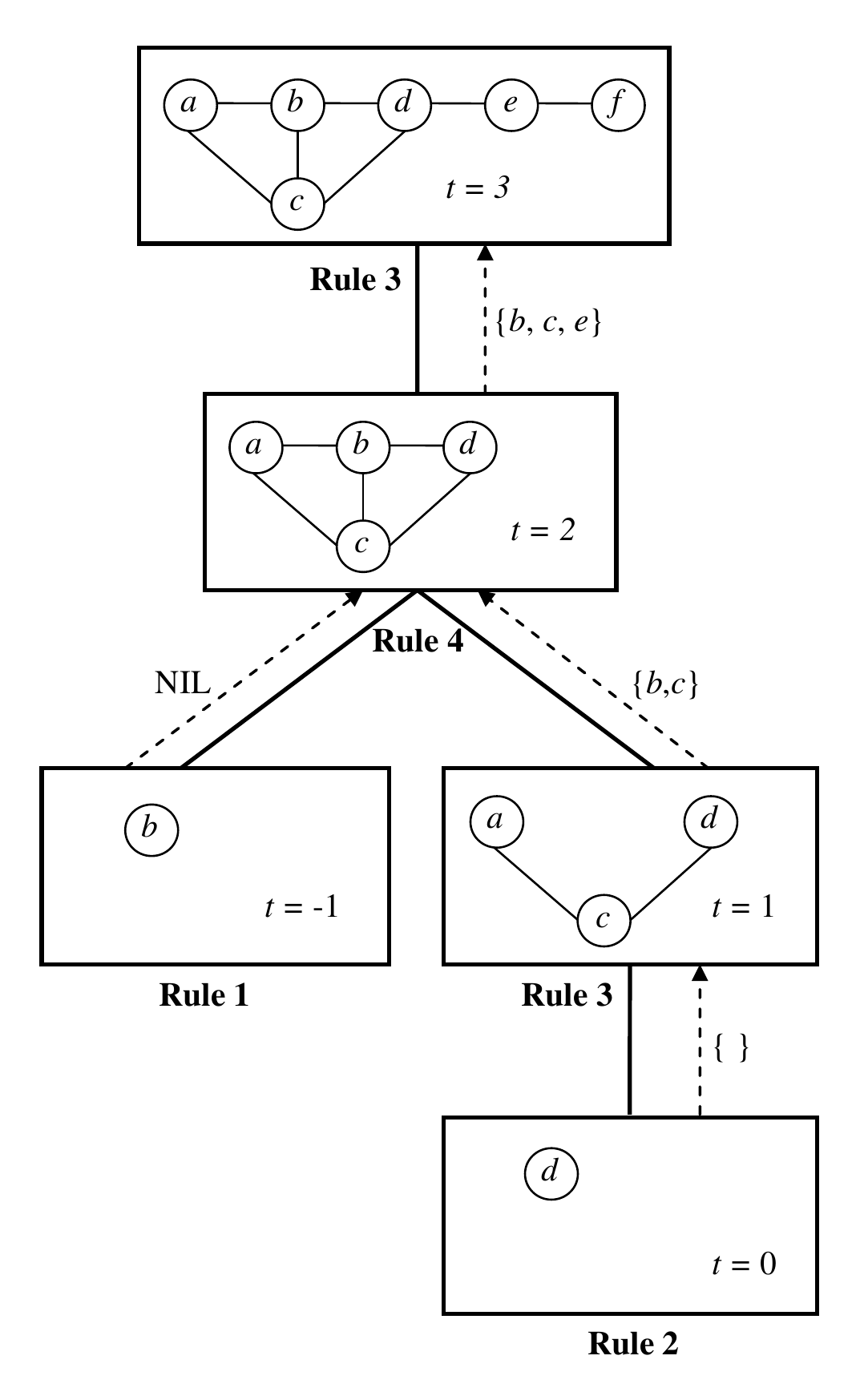}}
\caption{An illustration of a search tree corresponding to \alg{ALG1}.}
\label{fig:tree}
\end{figure}

\subsection{An Algorithm for {\sc WVC}}

We now show how to modify \alg{ALG1} in order to obtain an algorithm, \alg{ALG2}, for {\sc WVC}. Recall that, in the weighted variant of {\sc VC}, we have a weight parameter $W$ instead of the size parameter $t$; accordingly, we need to find a vertex cover of weight at most $W$. Therefore, we need to update Rule 1 to return NIL if $W < 0$, while Rule 2 remains correct. Therefore, we have the following reduction rules.

\setcounter{reducerule}{0}
\begin{reducerule}
{\normalfont [$W<0$]
Return NIL.}
\end{reducerule}

\begin{reducerule}
{\normalfont [$E=\emptyset$]
Return $\emptyset$.}
\end{reducerule}

{\noindent Rule 3, however, is more problematic. In terms of the size of the vertex cover, we have seen that it is always at least as good, when considering a leaf, to choose its neighbor. Similarly, in terms of weight, it is always at least as good, when considering a leaf whose weight is larger or equal to the weight of its neighbor, to choose the neighbor. This leads to the following reduction rule.}

\begin{reducerule}
{\normalfont [There is a leaf $v$ in $G$ such that $w(v)\geq w(u)$, where $u$ is the neighbor of $v$]
Return \alg{ALG2}$(G[V\setminus\{v,u\}],w,W\!-w(u))$, along with $u$.}
\end{reducerule}

{\noindent However, if the weight of the leaf is smaller than the weight of its neighbor, it is not clear whether we should choose the leaf or the neighbor. On the one hand, the leaf has a smaller weight; on the other hand, the neighbor, if it is not a leaf, covers more edges. Therefore, in the next rule, corresponding to Rule 4 in \alg{ALG1}, we cannot assume that $G$ does not contain leaves. However, $G$ contains a vertex of degree at least 2, since if it consists only of leaves and isolated vertices, then it contains two leaves that are neighbors, which invokes Rule 3 (the weight of one of them is smaller or equal to the weight of the other).}

\begin{branchrule}
{\normalfont [Remaining case] Let $v$ be a vertex of maximum~degree~in~$G$.
\begin{enumerate}
\item If the result of \alg{ALG2}($G[V\setminus N(v)],w,W\!-w(N(v))$) is not NIL: Return it along with~$N(v)$.
\item Else: Return \alg{ALG2}($G[V\setminus \{v\}],w,W\!-w(v)$), along with~$v$.
\end{enumerate}}
\end{branchrule}

{\noindent In this rule, the parameter $W$ is decreased by $|w(N(v))|\geq |N(v)|\geq 2$ in the first branch, and by $w(v)\geq 1$ in the second branch. Again, we have a branching vector that is at least as good as (2,1), whose root is smaller than 1.619. This implies that \alg{ALG2} runs in time $O^*(1.619^W)$.}

\subsection{An Algorithm for {\sc $k$-WVC}}

Finally, we modify \alg{ALG1} and \alg{ALG2} in order to obtain an algorithm, \alg{ALG3}, for {\sc $k$-WVC} (defined in Section \ref{section:technique}). Recall that, in this variant, we have a {\em special} size parameter $k$ {\em and} a weight parameter $W$. If there is a solution of size at most $k$ {\em and} weight at most $W$, we need to find a solution of weight at most $W$; otherwise, we need to find a solution of weight at most $W$ or return NIL. Therefore, we can update Rule 1 to return NIL if $\min\{W,k\} < 0$, while Rule 2 remains correct. Thus, we have the following reduction rules.

\setcounter{reducerule}{0}
\begin{reducerule}
{\normalfont [$\min\{W,k\}<0$]
Return NIL.}
\end{reducerule}

\begin{reducerule}
{\normalfont [$E=\emptyset$]
Return $\emptyset$.}
\end{reducerule}

{\noindent Now, by the arguments given for Rule 3 in \alg{ALG1} and Rule 3 in \alg{ALG2}, in terms of {\em both} weight and size, it is always at least as good, when considering a leaf whose weight is larger or equal to the weight of its neighbor, to choose the neighbor. This leads to the following reduction rule.}

\begin{reducerule}
{\normalfont [There is a leaf $v$ in $G$ such that $w(v)\geq w(u)$, where $u$ is the neighbor of $v$]
Return \alg{ALG3}$(G[V\setminus\{v,u\}],w,W\!-w(u),k-1)$, along with $u$.}
\end{reducerule}

{\noindent Next, by modifying Rule 4 of \alg{ALG1} and Rule 4  \alg{ALG2}, we can obtain a running time bounded by $O^*(1.619^k)$ (this modification, for completeness, is given below). However, suppose that we have devised rules that, if applied when $G$ does not contain leaves, have branching vectors whose roots are smaller than 1.619.\footnote{Indeed, in Section \ref{section:wvc1}, we show how to obtain such rules, assuming that we handle leaves by applying a reduction rule that is ``better'' than the one in this appendix.} We now show how, relying on our flexible use of the parameter $k$, we can apply a reduction rule that removes the remaining leaves from $G$ (i.e., after this rule, the graph $G$ does not contain leaves, and we can apply rules that rely on this assumption). To this end, consider a leaf, $v$, whose weight is smaller than the weight of its neighbor, $u$. As argued after Rule 3 in the previous subsection, it is not clear whether we should $v$ or $u$. We can sidestep this problem by manipulating the weight of $u$, as performed in the next rule, whose correctness is explained below.}

\begin{reducerule}
{\normalfont [There is a leaf $v$ in $G$ such that $w(v) < w(u)$, where $u$ is the neighbor of $v$]
Let the weight function $w'$ be defined as $w$, except for $w'(u)=w(u)-w(v)$. Return \alg{ALG3}($G[V\setminus \{v\}],w',W\!-w(v),k$), along with $v$ iff $u$ is not in the returned result.}
\end{reducerule}

{\noindent Consider the instance ($G[V\setminus \{v\}],w',W\!-w(v),k$), which appears in this rule. In this instance, on the one hand, choosing $u$ reduces $W\!-w(v)$ to $W\!-w(v)-w'(u)=W\!-w(u)$ and $k$ to $k-1$, which has the same effect as choosing $u$ in the original instance (i.e., in ($G,w,W,k$)). On the other hand,  {\em in terms of weight}, not choosing $u$ has the same effect as not choosing $u$ in the original instance: $W$ is reduced to $W\!-w(v)$ in ($G[V\setminus \{v\}],w',W\!-w(v),k$), and $v$ is added to the solution (since it is necessary to cover the edge $\{v,u\}$). Yet, {\em in terms of size}, not choosing $u$ has {\em almost} the same effect as not choosing $u$ in the original instance: the difference lies in the fact that we do not decrease $k$ by 1 (although $v$ is added to the solution). However, our flexible use of the parameter $k$ allows us to decrease its value by less than necessary: we may compute a vertex cover whose size is larger than $k$ (since we not decrease $k$ by 1), but we may not compute a vertex cover of weight larger than $W$. We could not simply call \alg{ALG3} with $k-1$, since then choosing $u$ overall decreases $k$ by 2, which is more than required (thus we may miss solutions by reaching Rule 1 too soon). Note that $w'(u)$ is positive, but might be smaller than $1$, which does not effect the correctness of the algorithm (in particular, the branching vector below relies on the change in the parameter $k$, rather than $W$).}

For the sake of completeness, we give below a straightforward adaptation of Rule 4 of \alg{ALG1} and Rule 4 of \alg{ALG2}, whose branching vector has a root smaller than 1.619 (with respect to $k$).

\begin{branchrule}
{\normalfont [Remaining case] Let $v$ be a vertex of maximum~degree~in~$G$.
\begin{enumerate}
\item If the result of \alg{ALG3}($G[V\setminus N(v)],w,W\!-w(N(v)),k-|N(v)|$) is not NIL: Return it along with~$N(v)$.
\item Else: Return \alg{ALG3}($G[V\setminus \{v\}],w,W\!-w(v),k-1$), along with~$v$.
\end{enumerate}}
\end{branchrule}

\section{Previous Work (Omitted Details)}\label{app:priorwork}

FPT algorithms for {\sc VC} in general graph, {\sc VC} in graphs of bounded degree 3, and {\sc IOB} are given in \cite{vc1993,vc1995,vc1998,vc1999a,vc1999b,vc2001,wvc2003,vc2005,vc2007,vc2010}, \cite{vc2001,3vc2000,3vc2005,3vc2009,3vc2010,3vc2013} and \cite{kISP24klogk,kIOB2klogk,kIOB49k,thesis11,kIOB16k,kISP8k,kISPbounddeg,ipec13,corrrepresentative}, respectively.

\renewcommand{\refname}{References in Appendices}

\end{document}